\documentclass[smallextended]{svjour3}

\usepackage{url}
\usepackage{amsmath}
\usepackage{tabularx}
\usepackage{algorithm}
\usepackage{algorithmicx}
\usepackage{algpseudocode}
\usepackage{graphicx}
\usepackage{subfigure}
\usepackage{epstopdf}
\usepackage{setspace}
\DeclareMathOperator*{\argmax}{\arg\max}

\newtheorem{myexp}{Example}

\begin{spacing}{1.0}
\begin{document}

\title{Diversification on Big Data in Query Processing}

\author{Meifan Zhang        \and
        Hongzhi Wang \and
        Jianzhong Li\and
        Hong Gao
}


\institute{Department of Computer Science and Technology, Harbin Institute of Technology\\
              \email{wangzh@hit.edu.cn}           %
}

\date{Received: date / Accepted: date}

\maketitle

\begin{abstract}
Recently, in the area of big data, some popular applications such as web search engines and recommendation systems, face the problem to diversify results during query processing. In this sense, it is both significant and essential to propose methods to deal with big data in order to increase the diversity of the result set. In this paper, we firstly define a set's diversity and an element's ability to improve the set's overall diversity. Based on these definitions, we propose a diversification framework which has good performance in terms of effectiveness and efficiency. Also, this framework has theoretical guarantee on probability of success. Secondly, we design implementation algorithms based on this framework for both numerical and string data. Thirdly, for numerical and string data respectively, we carry out extensive experiments on real data to verify the performance of our proposed framework, and also perform scalability experiments on synthetic data.
\end{abstract}


\section{Introduction}

	 Nowadays, as the amount of information dramatically increases in several popular applications, such as recommendation systems and web search engines, people are not only satisfied with relevant search results but also require that more relevant yet diverse topics are covered by a limited number of search results. Therefore, researchers pay considerable attention to the diversity of the search results. Many diversification methods specific to normal-scale data are proposed and applied in practice in order to improve the users' experience by avoiding the retrieval of too homogeneous results~\cite{Search_result_diversification}.
	
	 In big data era, the diversification of query results on big data is extremely important and meaningful for the following two reasons.
	
	 The first reason is that when dealing with big data, users can only access a small share of query results due to the huge amount. Thus, we have to return results of high quality. One of the proper indicators of results' quality is the diversity of the final result set. By providing users with query results diverse enough, we can help them obtain various and sufficient information in reasonable time and therefore, can improve users' satisfaction. For example, if a user wants to choose some detective stories, and he also wants to get the whole picture of this kind of books at the same time, then, he submits a search query of detective books to a book recommendation application. The expected result is a list of books written by different authors with different nationalities and published in a wide range of time rather than a list of books containing different editions of the same book. The diversity of the former list is obviously higher than that of the latter one. As the user wants to obtain a thorough knowledge of detective stories the book list with high diversity meets the demand of the user, which leads to a higher degree of user's satisfaction.
	
	 The second reason is that diversification methods on big data face several technical challenges and thus, is hard to be implemented. The memory is limited, but as a result of big data's essential characteristic-huge amount, we have difficulty in dealing with so much data in relatively small memory. This challenge prevents us from fully analysing big data and extracting diverse results from the big data to generate the final result set. Additionally, due to big data's huge amount, super-linear algorithms are not acceptable because big data can be scanned and processed only once or even less.
	
	 These two reasons above also make the diversification problem in big data hard to solve. In this sense, for big data, we should design linear or sub-linear space algorithms in order to retrieve diverse results in limited memory and also, we should design on-line algorithms for the restriction of scan times.

	While in query processing, many diversification methods have been proposed and implemented specific for normal-scale data~\cite{Efficient_diversity-aware_search,On_query_result_diversification}. However, they are not applicable for big data. This is because existing algorithms fail to take into consideration big data's huge amount and most of them have to scan the input data many times. Let us take the classical algorithm-Maximal Marginal Relevance (MMR) \cite{On_query_result_diversification} as an example. MMR incrementally selects an element from the candidate set and inserts it into the final result set. As for each incremental iteration, we have to scan the candidates once. Therefore, the MMR algorithm has to scan the input data for many times and is not suitable for big data.
	
	The query result diversification specialized in big data \linebreak[4] search has not been systematically studied so far, even though query result diversification over big data is meaningful and helpful in query processing. In the era of big data, for some popular applications, such as search engines and recommendation systems, we have to deal with a huge amount of on-line data, which can only be scanned no more than once. Therefore, it is both significant and essential to improve the diversity of results from big data within one scan.

This paper attempts to solve this problem. In this paper, we propose a stream-model-based diversification framework. As we know, this is the first one to diversify the results for big data within one scan. This framework is both effective and efficient in improving the diversity of the result set with low time complexity, low space cost and low computational overhead. Additionally, this framework processes the input data in on-line style and can be implemented for various data types concretely with its advantages.
	
	In this paper, we make the following contributions:
		
\begin{list}{\labelitemi}{\leftmargin=1em}\itemsep 0pt \parskip 0pt	
\item We propose a framework for diversification of query processing on big data. This framework allows us to improve the diversity of the final result set within one scan with guaranteed performance. In order to describe this framework thoroughly, we firstly present definitions of diversity, possible diversity gain (to describe an element's ability to improve overall diversity). Secondly, based on these definitions, the framework to solve the diversification problem on big data is proposed, formulated and evaluated. We prove that by assigning proper parameters in this framework, the probability of success in a single run can be guaranteed. As we know, this is the first paper to study the diversification method for big data within one scan.
	
\item On the basis of the proposed diversification framework, we design implementation algorithms for two common data types, numerical data and string data, to diversify the final returned results. We prove that the proposed framework can be implemented specific to different data types effectively and efficiently without degrading the framework's performance.
	
\item To verify the performance of proposed algorithms, we perform extensive experiments. From experimental results, our algorithms are verified to be effective and efficient, and we also study the factors that influence the performance of our methods and their influence to the performance. Then, in order to test the scalability of our proposed framework, we conduct experiments on synthetic data with a tremendous amount and study the framework's performance. Experimental results also demonstrate that the scalability and efficiency of our approach outperform existing approaches significantly.
\end{list}
	
	The rest of this paper is structured as follows. A brief review of related work about existing diversification methods is presented in Section~\ref{section: related work}. Also, Section~\ref{section: related work} demonstrates the intrinsic distinction between our proposed method and existing methods. Then in Section~\ref{section: framework}, we present a diversification framework which allows us to improve the diversity of the final result set within one scan with a guarantee of performance and effectiveness. In Section~\ref{section: implementation}, as for numerical data and string data, we present the basic definitions and the corresponding detailed implementation algorithms. Next, Section~\ref{section: experiment} illustrates the experimental results of our proposed methods on both real data and synthetic data, and then show thorough evaluations of them. Finally, Section~\ref{section: conclusions} concludes this paper and discusses some challenging yet interesting directions for further research.

\section{Related Work} \label{section: related work}
	Diversification in query processing has aroused many researchers' attention and interest these years. It can help enhance users' satisfaction in three aspects. Based on these aspects, diversification can be divided into three categories. The first category is content-based diversification, or similarity-based diversification \cite{Avoiding_monotony_improving_the_diversity_of_recommendation_lists}. It aims to return objects which are dissimilar enough to each other. The second category is intent-based diversification, or coverage-based diversification, which cover different aspects of a query, by returning objects from different categories or satisfying various possible users' intentions \cite{Diversifying_search_results}, \cite{Highlighting_Diverse_Concepts_in_Documents}, \cite{Bypass_rates_reducing_query_abandonment_using_negative_inferences}. The third category is novelty-based diversification, a method to return objects containing new information previously not seen before \cite{Novelty_and_diversity_in_information_retrieval_evaluation}, \cite{Novelty_and_redundancy_detection_in_adaptive_filtering}.

	 Drosou et al.~\cite{Search_result_diversification} categorize diversification methods into three parts, namely, content-based, intent-based and novelty-based diversification. Then they survey, compare and study the corresponding definitions, implementation algorithms of these three categories thoroughly and deeply.
	
	Vieira et al.~\cite{On_query_result_diversification} evaluate six existing content-based diversification methods thoroughly, and then propose two new approaches, namely, the Greedy with Marginal Contribution and the Greedy with Neighbourhood Expansion. Angel et al. \cite{Efficient_diversity-aware_search} focus on the problem of diversity-aware search which ranks results according to both relevance and dissimilarity to others. Then they propose a diversification method named DIVGEN which is qualified in terms of efficiency and effectiveness.
	
	As for intent-based diversification, researchers propose \linebreak[4] various methods in order to cover as many user intentions and topics as possible. In \cite{Improving_recommendation_lists_through_topic_diversification}, Ziegler et al. first propose a new metric, intra-list similarity, to represent the topical diversity of recommendation lists, and then a new method named topic diversification is proposed to decrease the intra-list similarity of lists. In \cite{Improving_personalized_web_search_using_result_diversification}, Radlinski et al. discuss three result diversification methods, namely, the Most Frequent method, the Maximum Result Variety method and the Most Satisfied method so as to improve personalized web search. Additionally, Capannini et al. \cite{Efficient_diversification_of_web_search_results} present an original method, OptSelect, in order to effectively and efficiently accomplish the diversification task.
	
	With regard to novelty-based diversification, not so many algorithms or implementation techniques have been put forward. Clarke et al. make a thorough distinction between novelty-the need to avoid redundancy-and diversity-the need to resolve ambiguity in \cite{Novelty_and_diversity_in_information_retrieval_evaluation}.

    Greedy algorithms are widely adopted to diversify the query results \cite{Improving_recommendation_lists_through_topic_diversification}, \cite{It_takes_variety_to_make_a_world_diversification_in_recommender_systems}, \cite{Efficient_computation_of_diverse_query_results}, they find the N elements first, and then test whether replacing an element in the origin set with a further element increases the diversity or not. Mostly, multiple passes of input data are needed in most of the methods.
	

	The essential difference between our proposed method and existing ones is that in our method, the input data only need to be scanned once or less, and during the scanning procedure, we select the element which will diversify the result set and put it into memory. On the contrary, existing diversification methods  mostly have to scan and process input data several times and also, cannot make effective use of available memory. In this sense, our presented method can solve the technical challenges which diversification tasks are faced with in big data, while existing methods fail to do so.

\section{General Framework of Diversification} \label{section: framework}
	In this section, we present a general framework of diversification for big data. Such framework is suitable for various data types and can effectively improve the overall diversity of the final result set. Firstly, we introduce some basic definitions to better establish the framework. Next we describe the proposed framework and explain how it works. Then, we perform evaluation of this framework. After this, we study the probability of success in improving the overall result set's diversity through this framework. Finally, we discuss the number samples to select during the implementation of this framework.

	As far as we are concerned, our proposed framework is the first one which offers users a diversification method for big data within one scan. It also has the advantage of low time complexity as well as low computational overhead. This proposed framework can also be implemented concretely for various data types without degrading its performance. In implementation, we often carry out randomly sampling to make it suitable for big data.

	As in our proposed method, we intend to scan the input big data no more than once and the data is scanned in order, we can logically consider the input big data as streaming data and therefore, can employ the data stream model, in which, the elements $a_1,a_2, \cdots, a_n$ have to be processed one after another. These elements in the stream model are not available for random access from disk or memory and can only be scanned once.

\subsection{Basic Definitions} \label{subsection: basic definitions}

\subsubsection{Definition of Diversity}
	We assume that the final result set with regard to a user query is limited and stored in memory.  We denote the size of available memory as $m$, the number of elements in the whole data set as $n$. Our task is to increase the diversity of the elements in memory. The definition of diversity is formalized as follows.
	
	As for several elements $x_1,x_2,\cdots,x_k$, which make up a vector denoted as $X$, the diversity of $X$ is computed as $Div(X)=Div(x_1,x_2,\cdots,x_k)$. Thus, the diversity of $m$ different elements $A=\{a_1,a_2,\cdots,a_m\}$ in memory is computed as $Div(A)=Div(a_1,a_2,\cdots,a_m)$. As for different data types, we choose various indicators for $Div(A)$.
	
	With regard to numerical data, we use variance to describe the diversity while as for string data, we express diversity as the sum of edit distances. Our definition of diversity is compatible with existing diversity definitions. In \cite{Search_result_diversification}, Drosou et al. mention a definition of diversity which interprets diversity as an instance of the $p$-dispersion problem. The $p$-dispersion problem is to pick out $p$ points from given $n$ points, so that the minimum distance between any two points is maximized \cite{A_comparison_of_p-dispersion_heuristics}. In this case, the definition of diversity is the minimum distance between any two points, while another existing definition is described as the average distance of any two points.
	
	Since the variance of a numerical data set can describe the distances between numerical elements in the set and similarly, the edit distance can represent the distance between each two strings, the proposed definitions of diversity are reasonable and also compatible with existing ones. The concrete definitions will be discussed in detail in Section~\ref{section: implementation}.

\subsubsection{Definition of Possible Diversity Gain}

	As the number of elements in available memory is fixed as a constant $m$, they are denoted as $A=\{a_1,a_2,\cdots,a_m\}$, where $A$ represents the original result set. As for an input element $\varphi$, we consider its contribution to the diversity of the result set. Although $\varphi$'s contribution to the overall diversity of the final result set could be defined in various ways, we decide to choose the change in diversity value caused when replacing $\varphi$ with another element in memory. This is because the difference of the ``new'' diversity when putting $\varphi$ in memory to take place of another element and the original diversity can both intuitively and clearly describes $\varphi$'s ability to diversify current result set. The detailed definition is formalized below.
	
	Given an element $\varphi$ in the input data, its \textit{possible diversity gain} to the diversity of the available memory is defined as $PDG(\varphi)$, which is computed as follows:
\begin{displaymath}
PDG(\varphi)=\max \limits_{a_j \in A}\{Div(A \setminus\{a_j\}\cup \{\varphi\})-Div(A)\}
\end{displaymath}

	From this formula, the possible diversity gain of a certain element $\varphi$ is the maximal difference of the diversity value, where the element replaces the $j$th element $a_j$ in the memory, and the memory's original diversity value.
	
	Also, it is obvious that $PDG(\varphi)$ is only related to the elements $A=\{a_1,a_2,\cdots,a_m\}$ in main memory, and has nothing to do with other input elements. If we consider the available memory as fixed, $PDG(\varphi)$ can be thought as an ``attribute'' of $\varphi$. Such attribute of $\varphi$ is helpful in our further discussion.

\subsection{Framework}
	As discussed before, we set the size of available memory as $m$, and set the number of elements in the whole data as $n$. In this part, we will describe our proposed diversification framework thoroughly and then explain how it works.
	
	Our proposed framework is an on-line algorithm, which scans the input data no more than once with a time complexity $O(n)$. Besides, the required space is far smaller than the input size. In this sense, the algorithm is a sub-linear-space algorithm. Also, this algorithm aims to select one element from input to replace it into available memory and therefore, can increase the diversity of the final result set stored in memory.
	
	 The pseudo-code illustrated in Algorithm~\ref{Diversify-Final-Result-Set} explains the procedure of our proposed framework. This framework increases the diversity of the results in main memory.

	 Firstly, when the scanning begins, $m$ different values are selected and stored into the available memory. We use the algorithm of counting distinct elements in streaming data \cite{Counting_distinct_elements_in_a_data_stream} to implement the procedure efficiently. Additionally, suppose the original maximal PDG value is $-{\infty}$, and this can be seen in Line 3 of  Algorithm~\ref{Diversify-Final-Result-Set}.
	
	Secondly, we apply the strategy of selecting a positive integer $k<n$, scanning the first $k$ elements of the input, then calculating their corresponding PDG values and storing the maximal PDG value $PDG_{max}$ among them. Line 4 to Line 9 in Algorithm~\ref{Diversify-Final-Result-Set} clearly demonstrate this procedure. Then the scanning of data continues. Line 10 to Line 16 in Algorithm~\ref{Diversify-Final-Result-Set} describe the following condition when the algorithm is executed. If an element $u$ with PDG value larger than $PDG_{max}$ among the following elements is met, the element $\argmax\limits_{u' \in A} Div(A \setminus \{u'\} \cup \{u\})$ in the memory will then be replaced by $u$ to maximize the difference of the diversity value.
	
	Thirdly, however, if no following elements have a larger PDG value than $PDG_{max}$, the last element in the data set, i.e., the $n$th element $u_n$ is used to take place of the element $\argmax\limits_{u' \in A} Div(A \setminus \{u'\} \cup \{u_n\})$ in the memory. This condition is demonstrated in Line 17 and Line 18 of Algorithm \ref{Diversify-Final-Result-Set}.
	
	 Through the procedure of this diversification framework, we can make sure which element from input data will be replaced into memory and used to increase the diversity of the final result set in memory.
	
\begin{algorithm}
\caption{Diversify-Final-Result-Set($k$,$n$)}\label{Diversify-Final-Result-Set}
\begin{algorithmic}[1]
\State \textbf{Input:} $A_0=\{a_1,a_2,\cdots,a_m\}$, $u_1,u_2,\cdots,u_n$
\State \textbf{Output:} final PDG value
\State $PDG_{\max} \leftarrow -{\infty}$
	\For {$i=1 \to k$}
   		\State $PDG(u_i)$ = Calculate-Element-PDG($u_i$)
   			\If {$PDG(u_i) > PDG_{\max}$}
   			\State $PDG_{\max} \leftarrow PDG(u_i)$
   			\EndIf
   \EndFor
   \For {$i=k+1 \to n$}
   		\State $PDG(u_i)$ = Calculate-Element-PDG($u_i$)
   			\If {$PDG(u_i) > PDG_{\max}$}
   			\State replace $\argmax\limits_{{u_i}' \in A_0} Div(A_0 \setminus \{{u_i}'\} \cup \{u_i\})$ with $u_i$
   			\State \textbf{return} $PDG(u_i)$
   			\EndIf
   \EndFor
\State replace $\argmax\limits_{{u_i}' \in A_0} Div(A_0 \setminus \{{u_i}'\} \cup \{u_n\})$ with $u_n$
\State \textbf{return} $PDG(u_n)$	
\end{algorithmic}
\end{algorithm}


	\textbf{Complexity Analysis:} Here we analyse the time complexity of the proposed diversification algorithm. Set the size of input as $n$ and suppose that the time complexity of the procedure  Calculate-Element-PDG($u_i$), which is invoked in Line 5 and Line 11 of Algorithm \ref{Diversify-Final-Result-Set}, is $h(m)$. This is a function of $m$, since the PDG value can be considered as an attribute of an input element and thus, is only related to the elements in memory.
	
	Firstly, during the scanning of first $k$ elements, we calculate each element's PDG value in order to initialize the $PDG_{\max}$ value. Therefore, it costs $k \cdot h(m)$ time. Then in the following scanning procedure, in Line 10 to Line 16 of Algorithm \ref{Diversify-Final-Result-Set}, the number of elements to be considered is no more than $n-k$. Therefore, the time complexity of this part is  O($(n-k) \cdot h(m)$). In a nutshell, the time complexity of this proposed diversification framework is O($n \cdot h(m)$). As the input size is $n$, and the size of available memory $m$ is a constant far smaller than $n$ and can be considered irrelevant to the time complexity, our proposed framework's time complexity is $O(n)$.	

	Even though this algorithm is in linear time, it indicates the worst case. Mostly we do not need to scan every element in n to increase the diversity. First, the elements in the memory will be scanned for $k$ times to get the $PDG\_max$. After that, we only need to find an element in the following elements which can increase the diversity with a replacement, and the rest of the following elements would not be scanned after the replacement. Mostly, the elements in the memory will be scanned for only a few times much smaller than n. The procedure of checking whether a certain element will be replaced into memory is performed online. Also, if the value of $k$ is chosen properly, we can guarantee the probability to successfully select the element with the maximal PDG value. This will be discussed more clearly in Section~\ref{subsection: probability of success}.

Additionally, this is a sub-linear space algorithm since the space spent is the input size's low-order function, and this algorithm can improve the diversity of elements stored in limited available memory. In this way, users can obtain more information in a fixed time period and thus, users' information needs can be satisfied more efficiently. As a result, the users' satisfaction is improved.

\subsection{The Probability of Success} \label{subsection: probability of success}
	In this section, we provide mathematical foundation of our proposed diversification framework. From the analysis, our framework is established to increase the overall diversity of the final result set by replacing an element with a relatively large PDG value from input into memory. In this sense, with the restriction that we only choose one element, if the element with the maximal PDG value is selected, we can increase the diversity to the utmost. Next, we will study the probability of this event's success.
	
	As stated in Section~\ref{subsection: basic definitions}, the proposed PDG value of a given element can be considered as an attribute of this element and thus, has nothing to do with the other input elements. In this sense, the PDG value here can be analogical to the score described in the on-line hiring problem~\cite{Introduction_to_algorithms}.
	
	Therefore, our mathematical foundation here uses the similar idea of the proof in \cite{Introduction_to_algorithms}. During the scanning of input data described in the proposed diversification framework, we hope to select the element with the maximal PDG value and replace it into memory. We denote this event as $U$, and the probability of this event's success as $Pr(U)$. Our task is to calculate $Pr(U)$ or obtain the range of $Pr(U)$.
	
	In the process of our derivation to obtain the mathematical foundation for computing $Pr(U)$, we first define several events which play an important role in this problem. Next, through derivation, we try to change the computational formula of $Pr(U)$ step by step. Then by simplifying the formula and solving some definite integrations, we finally get the range of $Pr(U)$.
	
	Firstly, we denote $U_i$ as the event that when the element with maximal PDG value is the $i$th element, we successfully select it and put in into memory. Then we obtain Theorem~\ref{thm:Pr(U) original} shown below.

\begin{theorem} \label{thm:Pr(U) original}
$Pr(U)=\sum \limits_{i=k+1} \limits^{n} Pr(U_i)$
\end{theorem}	

\begin{proof}
For various values of $i$, $U_i$ is non-intersect. Thus we can get $Pr(U)=\sum _{i=1} ^{n} Pr(U_i)$. According to the description of our algorithms, we will fail to select the best element, i.e., the element with the maximal PDG value, if this element appears in the first $k$ positions from input. Thus, $Pr(U_i)=0$, where $i=1,2, \cdots, k$. In this way, we get Theorem~\ref{thm:Pr(U) original}.
\end{proof}

	Now we aim to calculate $Pr(U_i)$. Assume the input data is denoted as $u_1,u_2, \cdots, u_n$. Set $M(u_j)=\max \limits_{1 \leq i \leq j} \{PDG(u_i)\}$ as the highest PDG value of the element which is among $u_1,u_2, \cdots, u_j$. If the event $U_i$ is to take place, two other events have to happen simultaneously. One is that the element with the highest PDG value must be in the position of the element $u_i$, and we define this event as $R_i$. The other is that all the PDG values of elements among $u_{k+1},u_{k+2}, \cdots, u_{i-1}$ must be smaller than $M(u_k)$. This event is denoted as $T_i$. Then we can obtain Theorem~\ref{thm:Pr(Ui)} demonstrated below.

\begin{theorem}\label{thm:Pr(Ui)}
$Pr(U_i)=Pr(R_i \cap T_i)=Pr(R_i) \cdot Pr(T_i)$
\end{theorem}

\begin{proof}
The reason why the event $R_i$ has to happen to make sure that $U_i$ is to take place is obvious. Then, only if $T_i$ happens, the proposed algorithm will not select the element from the ($k+1$)th element to the ($i-1$)th element. $R_i$ and $T_i$ are two independent events, so we obtain the expression of $Pr(U_i)$ according to the statistics property of independent events.
\end{proof}

	Then the range of the probability of success $Pr(U)$ can be obtained and illustrated in Theorem~\ref{thm:probability_of_success}.

\begin{theorem}\label{thm:probability_of_success}
$\dfrac{k}{n}(\ln n - \ln k) \leq Pr(U) \leq \dfrac{k}{n}(\ln (n-1)-\ln (k-1))$
\end{theorem}

\begin{proof}
	We can easily figure out that $Pr(R_i)=1/n$ as the maximal PDG value is equally likely to appear in $n$ positions of input. Then $Pr(T_i)=k/(i-1)$, because if $T_i$ happens, it means that the highest PDG value among $u_1,u_2, \cdots,u_{i-1}$ must appear in the first $k$ positions of input and also, this value is equally likely to appear in the $k$ positions. Therefore, $Pr(U_i)=k/(n(i-1))$. Then we can finally get the computational formula of $Pr(U)$:

\begin{displaymath}
Pr(U)=\sum \limits_{i=k+1} \limits^{n} Pr(U_i)=\sum \limits_{i=k+1} \limits^{n} \dfrac{k}{n(i-1)}
\end{displaymath}

After simplification, the formula above is expressed as:

\begin{displaymath}
Pr(U)=\dfrac{k}{n}\sum \limits_{i=k+1} \limits^{n} \dfrac{1}{i-1}=\dfrac{k}{n}\sum \limits_{i=k} \limits^{n-1} \dfrac{1}{i}
\end{displaymath}

	Then we can use integration to make constraints on the upper bound and lower bound of $\sum_{i=k} ^{n-1} (1/i)$ as follows:
\begin{displaymath}
\int_k^n \dfrac{1}{x}dx \leq \sum \limits_{i=k} \limits^{n-1} \dfrac{1}{i} \leq \int_{k-1}^{n-1} \dfrac{1}{x}dx
\end{displaymath}
After solving the definite integrations above, we get the final upper bound and lower bound of $Pr(U)$ demonstrated in Theorem~\ref{thm:probability_of_success}.
\end{proof}

	In summary, the probability of successfully choosing the element with maximal PDG value in one experiment $Pr(U)$ will be no less than $(k/n) \cdot (\ln n - \ln k)$ and no more than $(k/n) \cdot (\ln (n-1)-\ln (k-1))$. Both upper bound and lower bound are related to the value of $k$ and $n$. In a certain experiment, the value of $n$ is set. Then if $k$ is chosen properly, we can guarantee the probability of success in one single experiment is satisfactory.

\subsection{Discussion} \label{subsection: discussion}
	As described before, we know how the proposed framework works and why this framework has a guarantee on probability of success. In the implementation of this diversification framework, due to the large amount of big data, we have to pick out representative samples from the input data instead of accessing the whole data set. In this section, we will discuss how to select a proper number of samples.
	
	As the total amount of input data can be large, sampling is essential in experiments. Hence, we divide the input into a certain number of segments with the same size. These segments can be considered as the samples of the whole big data set. Concretely, we try to divide the data file into segments with a fixed size of $a$, then the whole data file is partitioned into $\lfloor n/a \rfloor$ segments. We select $s$ segments from them and use them as experimental samples. As we carry out such sampling experiments and obtain corresponding results based on the samples selected, it is apparent that variables $s$ and $a$ will exert an influence on the performance.
	
	Now suppose we select $n$ samples to represent the whole data set, and we set $X_i$ to represent the result of the $i$th sample. $X_i=1$ represents that we manage to find the element with the maximal PDG value and replace it into memory in order to increase the diversity. Then, $X_i=0$ means that we fail to do so. It is obvious that $X_1, X_2, \cdots, X_n$, the $n$ samples, can be treated as a sequence of independent Poisson tests. $X_1,X_2, \cdots, X_n$ all satisfy that $Pr(X_i=1)=P_i$, where $i=1,2, \cdots, n$. Then we suppose that $X=\sum \limits_{i=1} \limits^n X_i$ and then $\mu = E(X)$, using $\mu$ to denote the mathematical expectation of $X$. As a result of the properties of the expectation and the summation, we obtain the following equation.

\begin{displaymath}
\mu = E(X)=E[\sum \limits_{i=1} \limits^n X_i]=\sum \limits_{i=1} \limits^n E[X_i]=\sum \limits_{i=1} \limits^n P_i
\end{displaymath}

	Now we consider how to measure the effectiveness and reliability of our proposed diversification framework when using samples to represent the whole input. From \cite{Probability_and_computing:_Randomized_algorithms_and_probabilistic_analysis}, we have the following theorem.

\begin{theorem}\label{thm:probability_of_success_sample}
If $X_1,X_2, \cdots, X_n$ are independent Poisson tests and they satisfy that $Pr(X_i)=P_i$, set $X=\sum \limits_{i=1} \limits^n X_i$, $\mu=E[X]$, then as for $\forall 0< \delta <1$, the inequality below holds.
\begin{displaymath}
Pr(|X-\mu| \geq \delta \mu) \leq 2e^{-\mu \delta ^2 /3}
\end{displaymath}
\end{theorem}	

	According to Theorem~\ref{thm:probability_of_success_sample}, we know that the effectiveness and reliability of our sampling method is closely related to a function of $\mu$. Thus we define $f(\mu)=2e^{-\mu \delta ^2 /3}$ to describe it. Then let us provide a more detailed discussion on $f(\mu)$. The 1st derivative of $f(\mu)$ can be easily calculated as $f'(\mu)=2e^{-\mu \delta ^2 /3} \cdot (- \dfrac{\delta ^2}{3})$. Clearly, $f'(\mu) \leq 0$. It means that $f(\mu)$ decreases gradually with $\mu$ increasing. Thus, suppose the maximal value of $\mu$ is $\mu _{max}$ and the minimal value of $\mu$ is $\mu _{min}$. Then as a result of $f(\mu)$'s properties, we make sure that $f(\mu _{max}) \leq f(\mu) \leq f(\mu _{min})$. From the inequality presented in Theorem~\ref{thm:probability_of_success_sample}, we can get
\begin{displaymath}
	Pr(|X-\mu| \geq \delta \mu) \leq f(\mu),
\end{displaymath}
so we can get
\begin{displaymath}
Pr(|X-\mu| \geq \delta \mu) \leq f(\mu _{min}).
\end{displaymath}

	Now our task is to compute $f(\mu _{min})$. As $\mu =\sum \limits_{i=1} \limits^n P_i$, the minimal value of $\mu$ is presented as $\mu _{min}=n \cdot \min \limits_{1 \leq i \leq n} P_i$. Back in the inequality in Theorem~\ref{thm:probability_of_success} in Section~\ref{subsection: probability of success}, we get the following formula.
\begin{displaymath}
\dfrac{k}{n}(\ln n - \ln k) \leq Pr(U) \leq \dfrac{k}{n}(\ln (n-1)-\ln (k-1)).
\end{displaymath}
Here $Pr(U)=\dfrac{k}{n}\sum \limits_{i=k} \limits^{n-1} \dfrac{1}{i}$ represents the probability that one single experiment is successful. In this sense, $Pr(U)$ is equivalent to $P_i$. Thus, it is true that
\begin{displaymath}
\dfrac{k}{n}(\ln n - \ln k) \leq P_i \leq \dfrac{k}{n}(\ln (n-1)-\ln (k-1)).
\end{displaymath}

	If we try to calculate $\mu_{min}$, we have to compute the minimal value of $P_i$. Set $h(k)=\dfrac{k}{n}(\ln n - \ln k)$, as $P_i \geq h(k)$, then it is apparent that $P_i \geq h(k)_{max}$, and the minimal value of ${P_i}$ is $h(k)_{max}$. After calculation and analysis, when $k=\dfrac{n}{e}$, the value of $h(k)$ is the largest and $h(k)_{max}=\dfrac{1}{e}$.
	
	Therefore, $\mu _{min}=n \cdot \dfrac{1}{e}=\dfrac{n}{e}$. Then we can get the following theorem.
\begin{theorem} \label{thm:finalinequality}
$Pr(|X-\mu| \geq \delta \mu) \leq 2e^{- \dfrac{\delta ^2 \cdot n}{3e}}$
\end{theorem}
	
	This inequality measures the degree that the results of samples deviate from the mathematical expectation of the population. Theorem~\ref{thm:finalinequality} describes the probability that the deviation of this sampling result from the expectation of the population is no less than $\delta$ times of $\mu$. We denote its right section as $p_0$. Therefore, we can get $Pr(|X-\mu| \geq \delta \mu) \leq p_0$.

	 As $\delta$ is an indicator of the deviation degree of sampling results from the expectation of the population, different values of $\delta$ have various mathematical meanings. $n$ describes the number of picked out samples. By choosing various values of $\delta$ and $n$, we can get different $p_0$. That is, diverse probability values.
	
	
	After testing the change of $p_0$ values with $\delta$ and $n$, two conclusions are drawn. The first is that with a given $n$, if $\delta$ increases, then the value of $p_0$ correspondingly decreases. The second conclusion is that when the value of $\delta$ stands constant, with $n$ increasing, the value of $p_0$ is on the decrease.
	
	Therefore, it is learned that given a certain $\delta$, which represents the deviation degree, we should correspondingly choose a proper value of $n$, i.e., the number of samples, in order to make sure that the value of $p_0$ is small enough. Only with this method, can we ensure the effectiveness and reliability of our diversification framework using samples to represent the whole population.
\section{Implementation Algorithms} \label{section: implementation}
	In Section 3, we thoroughly describe the diversification \linebreak[4]framework and then explain how it works. In this section, we develop implementation algorithms based on our proposed framework for two kinds of data types, numerical data and string data. Note that these two kinds covers most data types. For instance, both integer and double values can be processed as numerical data while both text and category attributes can be processed as string data.
\subsection{Algorithms for Numerical Data}
\subsubsection{Expression of Diversity and PDG}
	In this section, we assume the input data is in the form of numerical values. Then our task is to pick out a portion from the query results on a massive numerical data set and then improve the diversity of this portion.

	As discussed in Section 3, we assume that the final result set with regard to a user query is stored in memory. Then we denote the size of available memory as $m$, the number of elements in the whole data set as $n$, and our task is to increase the diversity of the elements in memory.
	
	As for several numerical values $X=\{x_1,x_2,\cdots,x_k\}$, we describe the diversity $Div(X)$ of $X$ as the variance, represented as $Div(X)=Var(X)=E[(X-\mu)^2]$, where $\mu$ is the average value of $X$. Thus, the diversity of $m$ different values $A=\{a_1,a_2,\cdots,a_m\}$ in memory can be computed as $Div(A)=Div(a_1,a_2,\cdots,a_m)=\dfrac{1}{m} \cdot \sum\limits_{i=1}\limits^m (a_i-\mu_a)^2$, where $\mu_a$ is the average of $a_1,a_2,\cdots,a_m$.

	Then, given an element $\varphi$ in the input, we describe its \textit{possible diversity gain} to the diversity of the available memory as $PDG(\varphi)$. In this case, $PDG(\varphi)$ is computed as follows.

\begin{displaymath}
PDG(\varphi)=\max \limits_{a_j \in A}\{Var(A \setminus\{a_j\}\cup \{\varphi\})-Var(A)\}
\end{displaymath}
	
	Variance measures how far the elements are spread out from the mean value and from each other. Therefore, if the variance value of several elements is large, then it means that these elements are pretty different from each other. It correspondingly means that the diversity of this element set is large. In this sense, variance is a proper indicator of diversity among numerical elements. Additionally, with our simple but useful definition of diversity, it can really decrease the computational overhead of our proposed method.
\subsubsection{Implementation Algorithm}
	As demonstrated in Algorithm~\ref{Diversify-Final-Result-Set} in Section 3, the algorithm invokes the function $PDG(u_i)$ = Calculate-Element-PDG($u_i$) during its execution. This algorithm specific to numerical data is shown in Algorithm~\ref{Calculate-Element-PDG-1}.
\begin{algorithm}
\caption{Calculate-Numerical-Element-PDG($\varphi$)}\label{Calculate-Element-PDG-1}
\begin{algorithmic}[1]
\State \textbf{Input:} values in main memory $A_0=\{a_1,a_2, \cdots, a_m\}$, double element $\varphi$
\State \textbf{Output:} PDG value of $\varphi$
\State calculate the variance of $A_0$ as $result$
   \For {$i=1 \to m$}
   		\State calculate the variance of $A'$ where $\varphi$ takes place of $a_i$ as $variance_i$
   		\If {$variance_i > result$}
   		\State $result \leftarrow variance_i$
   		\EndIf
   \EndFor
\State $PDG(\varphi) \leftarrow result$
\State \textbf{return} $PDG(\varphi)$
\end{algorithmic}
\end{algorithm}

Here we analyse the procedure of Algorithm~\ref{Calculate-Element-PDG-1} in detail. The elements in memory are described as a vector $A_0=\{a_1,a_2, \cdots, a_m\}$. As for an input element $\varphi$, our intention is to output $PDG(\varphi)$. As illustrated in Line 3, we first compute the variance of $A_0$ and assign this value to $result$. Then in the loop presented in Line 4 to Line 9, we take into account the ``new'' variance of the result set when $\varphi$ takes place of each element stored in memory. Within the loop, if this ``new'' variance is larger than $result$, then we update the value of $result$ with this variance value. When the loop is terminated, in Line 10, the value of $result$ is assigned to $PDG(\varphi)$ and thus, this algorithm halts.

	\textbf{Time Complexity Analysis:} As in implementation, we simplify the formula of variance $\dfrac{1}{m} \cdot \sum\limits_{i=1}\limits^m (a_i-\mu_a)^2$ to $\dfrac{1}{m} \cdot (\sum\limits_{i=1}\limits^ma_i^2-m\mu_a^2)$. Therefore, it takes only $\theta(1)$ time to calculate the new variance if we replace an input element with a certain one in memory. In Algorithm \ref{Calculate-Element-PDG-1}, we execute $m$ times of the calculation process described above, so the time complexity of Algorithm~\ref{Calculate-Element-PDG-1} is $\theta (m)$.

We then use an example below to illustrate the algorithm.

\begin{myexp}
Consider the following scenario. A user wants to select a house to purchase from a tremendous \linebreak[4]amount of information, but the concrete intention for the spot or the size is not clear. Then we can really do him a favour by providing him with a certain number of houses with diverse sizes and layouts, which of course, are in the form of numerical values. Let us just take this situation as an example.
	
	Suppose there are 25 numerical elements in the input, \linebreak[4]which successively are 711.56, 121.65, 7498.12, 2866.83, \linebreak[4]794.47, 7638.57, 9561.95, 6819.74, 8324.07, 2753.54, -272.60, \linebreak[4]3396.49, 3857.34, 5266.30, 2788.52, 4681.03, 6.34, 5494.43, -8914.71, 7603.40, 1428.25, 591.98, 3332.02, 9255.67, \linebreak[4]7133.70. Set the size of available memory as $m=5$, then the parameter $n$ which means the number of input elements described in the framework is correspondingly 20. We assume the number of elements to scan is $k=10$. Thus, for 25 numerical elements, 5 will be stored in memory and the following $n=20$ will be used as input data.
	
	From Algorithm~\ref{Calculate-Element-PDG-1}, firstly, we choose $m$ different elements into memory and they are 121.65, 711.56, 794.47, 2866.83, 7498.12 in order. Secondly, we try to scan the $n$ elements in order. We scan the next $k$ elements and find out the $PDG_{\max}$ value is from the element 9561.95, which is in the second position of $n$ elements. Then the scanning continues. When we get to the 14th element -8914.71, it is found that its PDG value is larger than the $PDG_{\max}$ value and it needs to replace the second element in memory. Thirdly, we carry out the replacement and improve the diversity of elements stored in memory.
	
	However, if we change the element -8914.71 in input data to 8914.71 and then also carry out the procedure of Algorithm~\ref{Calculate-Element-PDG-1}, we may find the implementation results are quite different. The $PDG_{\max}$ value stays the same, but when the scanning continues, no following elements have a larger PDG value than $PDG_{\max}$. Thus, we select the last element 7133.70 to replace into memory and it needs to take place of the 4th element in it.
	
	There are two possible execution results of Algorithm~\ref{Calculate-Element-PDG-1}. With the implementation of this algorithm, the consumer access enough measurements and thus, can have more options. In this sense, our work is quite meaningful.
\end{myexp}	
\subsection{Implementation on String Data}
\subsubsection{Expression of Diversity and PDG}
	To study how to describe the diversity of a string set, we first lay emphasis on how to represent the dissimilarity between two strings. Here, we use edit distance, which is the minimum operations from insertion, deletion or substitution to change one string to another, so as to measure the difference between two strings. We choose edit distance because it is often used to measure the difference between strings \cite{Efficient_Approximate_Search_on_String_Collections}. Then given two strings $x$, $y$, the dissimilarity between them is denoted as $Dis(x,y)=EditDis(x,y)$.
	
	Thus, we represent the diversity of a string set $S=\{s_1,s_2,\cdots,s_m\}$, whose cardinality is $m$, as follows:
\begin{displaymath}
Dis(S)=\sum \limits_{i=1} \limits^{m-1} \sum \limits_{j=i+1} \limits^{m} Dis(s_i,s_j)=\sum \limits_{i=1} \limits^{m-1} \sum \limits_{j=i+1} \limits^{m} EditDis(s_i,s_j)
\end{displaymath}
Note that $Dis(S)$ is equivalent to $Div(S)$ when dealing with string data.

	Then, we describe a string element's PDG quantitatively. Given a certain string element $\varphi$ in the input data, we define its \textit{possible diversity gain} as follows.
\begin{displaymath}
PDG(\varphi)=\max \limits_{s_j \in S}\{Dis(S \setminus \{s_j\} \cup \{\varphi\})-Dis(S)\}
\end{displaymath}

	Edit distance can well measure the difference between two strings. In this sense, we sum up the edit distances between all pairs of strings in a string set together, and then use the result to describe how various the string elements in this set are. Therefore, this sum result denoted as $Dis(S)$ can reasonably represent the diversity of this string set $S$.
\subsubsection{Implementation Algorithm}
	We have proposed the implementation of the algorithm named $PDG(u_i)$ = Calculate-Element-PDG($u_i$) specific to numerical data in Algorithm~\ref{Calculate-Element-PDG-1}. Here we focus on the implementation of our proposed diversification framework specific to string data. Its pseudo-code is demonstrated in Algorithm~\ref{Calculate-Element-PDG-2}.
\begin{algorithm}
\caption{Calculate-String-Element-PDG($\varphi$)}\label{Calculate-Element-PDG-2}
\begin{algorithmic}[1]
\State \textbf{Input:} available memory $A_0=\{a_1,a_2, \cdots, a_m\}$, string element $\varphi$
\State \textbf{Output:} PDG value of $\varphi$
\State calculate the edit distance sum of $A_0$ as $result$
   \For {$i=1 \to m$}
   		\State calculate the edit distance sum of $A'$ where $\varphi$ takes place of $a_i$ as $distance_i$
   		\If {$distance_i > result$}
   		\State $result \leftarrow distance_i$
   		\EndIf
   \EndFor
\State $PDG(\varphi) \leftarrow result$
\State \textbf{return} $PDG(\varphi)$
\end{algorithmic}
\end{algorithm}

	Now we explain the procedure of Algorithm~\ref{Calculate-Element-PDG-2} in detail. At first, strings $a_1,a_2, \cdots, a_m$ are stored in memory and form a vector denoted as $A_0$. The input of this algorithm is a string element $\varphi$ and the output is $PDG(\varphi)$. In Line 3, we firstly calculate the sum of edit distances between each two strings in $A_0$ and assign this value to $result$. Secondly, in the loop in Line 4 to Line 9, we consider the conditions where $\varphi$ takes place of each element in memory, and then compute the sum of edit distances, respectively. With this value, we then check whether it is larger than current $result$ value. If so, we assign this edit distance sum value to $result$. Finally, in Line 10, $PDG(\varphi)$ is set as $result$ and then the algorithm is terminated.

	\textbf{Time Complexity Analysis:} Suppose the length of \linebreak[4]strings considered is $a$. Then it takes $\theta (a^2)$ time to calculate the edit distance between two strings. In implementation, we employ a method that only takes into account the calculation of the edit distances between the strings which are influenced by the input string $\varphi$. Then as to one replacement involving $\varphi$, we execute $\dfrac{1}{2} m(m-1)$ operations of calculating the edit distance between two strings, so the time complexity of dealing with one element is $\theta (m \cdot a^2)$. Since we have to process $m$ replacements, the overall time complexity of Algorithm~\ref{Calculate-Element-PDG-2} is $\theta(m^2a^2)$. As $a$ is irrelevant to $m$ and is fixed when the input data file is determined, then the time complexity is $\theta(m^2)$.
We now take a practical situation as an example and explains the procedure of the implementation method more thoroughly.
	
\begin{myexp}
	Here we collect the input data from the website \url{http://www.amazon.com/}. We choose the category \linebreak[4]``Books'' and then focus on the book list called \textit{``100 Books to Read in a Lifetime: Readers' Picks''}. We randomly choose several books and use their names as input in this example.
	
	Suppose the input has 15 strings. That is $s_1$: ``A Brief History of Time'', $s_2$: ``Alice Munro: Selected Stories'', $s_3$: ``Bel Canto'', $s_4$: ``Charlie and the Chocolate Factory'', $s_5$: ``Daring Greatly: How the Courage to Be Vulnerable Transforms the Way We Live, Love, Parent, and Lead'', $s_6$: ``Great Expectations'', $s_7$: ``Harry Potter and the Sorcerer's Stone'', $s_8$: ``Invisible Man'', $s_9$: ``In Cold Blood'', $s_{10}$: ``Jimmy Corrigan: Smartest Kid on Earth'', $s_{11}$: ``Kitchen Confidential'', $s_{12}$: ``The Devil in the White City: Murder, Magic, and Madness at the Fair that Changed America'', $s_{13}$: ``Love in the Time of Cholera'', $s_{14}$: ``Man's Search for Meaning'', $s_{15}$: ``The Lion, the Witch and the Wardrobe''.
	
	Set $m=5$, $k=5$ and $n=10$. Firstly, Algorithm~\ref{Calculate-Element-PDG-2} scans the input data and store $m$ different elements in memory, and the string elements are $s_1$, $s_2$, $s_3$, $s_4$, and $s_5$ in order. Then we scan the next $k$ string elements and pick out the $PDG_{\max}$ value which is from $s_{10}$. Then the scanning continues. When we reach $s_{12}$, this element has a larger PDG value than $PDG_{\max}$ and it needs to take place of the second element, that is, $s_2$ in memory. Then the replacement is carried out and the overall diversity of available memory is increased.
	
	Nevertheless, if we change $s_{12}$ in input file to ``Life After Life'' and still execute the algorithm, we will find that $PDG_{\max}$ will not vary, but no following elements have a larger PDG than $PDG_{\max}$. Then the scanning continues and the last element, $s_{15}$, is met. The last element should replace $s_1$ in memory to obtain its PDG. Therefore, we switch $s_{15}$ with $s_1$ in memory and manage to diversify the results.
	
	There are two possible conditions of this method and they are both illustrated in the example above. Additionally, in this scenario, to choose a book from the given book list, our proposed method allows diversifying the returned book titles and thus, can improve users' satisfaction and experience.
\end{myexp}
\section{Experimental Evaluation} \label{section: experiment}
	In this section, we evaluate our proposed methods with extensive experiments.
	
	\textbf{Test Dataset: } We choose various data sets according to the features of two different data types: as for numerical data, we use double data extracted from TPCH \cite{TPCH_double_data}. The scale of double data is 150,000.
	
	With regard to string data, we carry out our experimental evaluation on DBLP dataset, which is available in \cite{DBLP_string_data} and the scale of string data utilized is 500,000. The titles of various papers are used as the input of experimental evaluation of string data.
\subsection{Experimental Methods}
	\textbf{Independent variables:} In our algorithms, three parameters affect the effectiveness and efficiency of our methods. The first one is the size of available memory $m$. The second parameter is the total number of input elements $n$ and the third parameter is the number of elements $k$ to be firstly scanned in the procedure.
	
	As discussed in Section~\ref{subsection: discussion}, the cardinality of selected samples $a$ and the number of selected samples $s$ both exert an influence on our methods' performance. Hence, these two parameters have to be taken into consideration as well.
	
	As discussed above, in the global experiments, we consider five parameters, $m$, $n$, $k$, $a$ and $s$. We study and analyse these their influence on the efficiency and effectiveness of proposed methods, respectively. These five parameters may affect each other to some extent, so we employ the control variate method to study the impact of each parameter on the proposed methods' performance.
	
	\textbf{Dependent variable:} In order to better measure the impact of these variables on returned results' diversity, we present a measurement to describe the degree that the diversity is increased. With the definition of \textit{possible diversity gain} (PDG) value proposed in section 3 and section 4, suppose the original diversity value computed is $div_0$, and the PDG value after diversification is denoted as $PDG$, we define \textit{diversity increasing rate}, i.e., \textit{DIR value} as $DIR=\dfrac{PDG}{div_0}$.
	
	In the next two subsections, we study and analyse the impact of five parameters, $m$, $n$, $k$, $a$ and $s$ on $DIR$ value and discuss about whether and how they affect the performance of our diversification methods.
	
	\textbf{About $\delta$:} When we try to study the influence of five control parameters on the diversity increasing rate, we base our experimental results on a fixed number of randomly selected samples. Note that we employ the simple random sampling scheme without replacement. In most of our presented figures (except for the experimental study on the impact of selected sample number $s$), three separate lines are illustrated with the legends named respectively as $\delta=0.2$, $\delta=0.5$, and $\delta=0.8$, respectively. $\delta$ describes the degree that sampling results deviate from population results.
	
	\textbf{About $p_0$:} Considering the meaning of the formula discussed in Section~\ref{subsection: discussion}, $Pr(|X-\mu| \geq \delta \mu) \leq p_0$, these three lines are drawn using various $\delta$ and different $s$ in order to void large $p_0$. As for $\delta=0.2$, we set $s=500$, and then we get $p_0=1.72151 \times 10^{-1}$, i.e., $Pr(|X-\mu| \geq \delta \mu) \leq 1.72151 \times 10^{-1}$. This means that the probability that, our sampling results deviate from the theoretical results of all data, to a degree of more than $\delta$ times of expectation value, is less that $1.72151 \times 10^{-1}$. As the value of $p_0$ is small enough, we can guarantee that the results of our sampling experiments are credible.
	
	Similarly, with regard to $\delta=0.5$ and $\delta=0.8$, we respectively set the value of $s$ as 300 and 100, and then the value of $p_0$ is correspondingly $2.02689 \times 10^{-4}$ and $7.80991 \times 10^{-4}$. As the values of $p_0$ are small enough, we ensure that the experimental results of these two scenarios are both trustable.
\subsection{Experiments on Numerical Data} \label{subsection: experiment double}
	In this section, we analyse the impact of the five parameters on the DIR value, when dealing with numerical data and aiming to diversify returned results.
\subsubsection{Impact of Memory Size}
\begin{figure}
\centering
\subfigure[On DIR]{
\label{double-m-dir}
\includegraphics[width=0.4\textwidth]{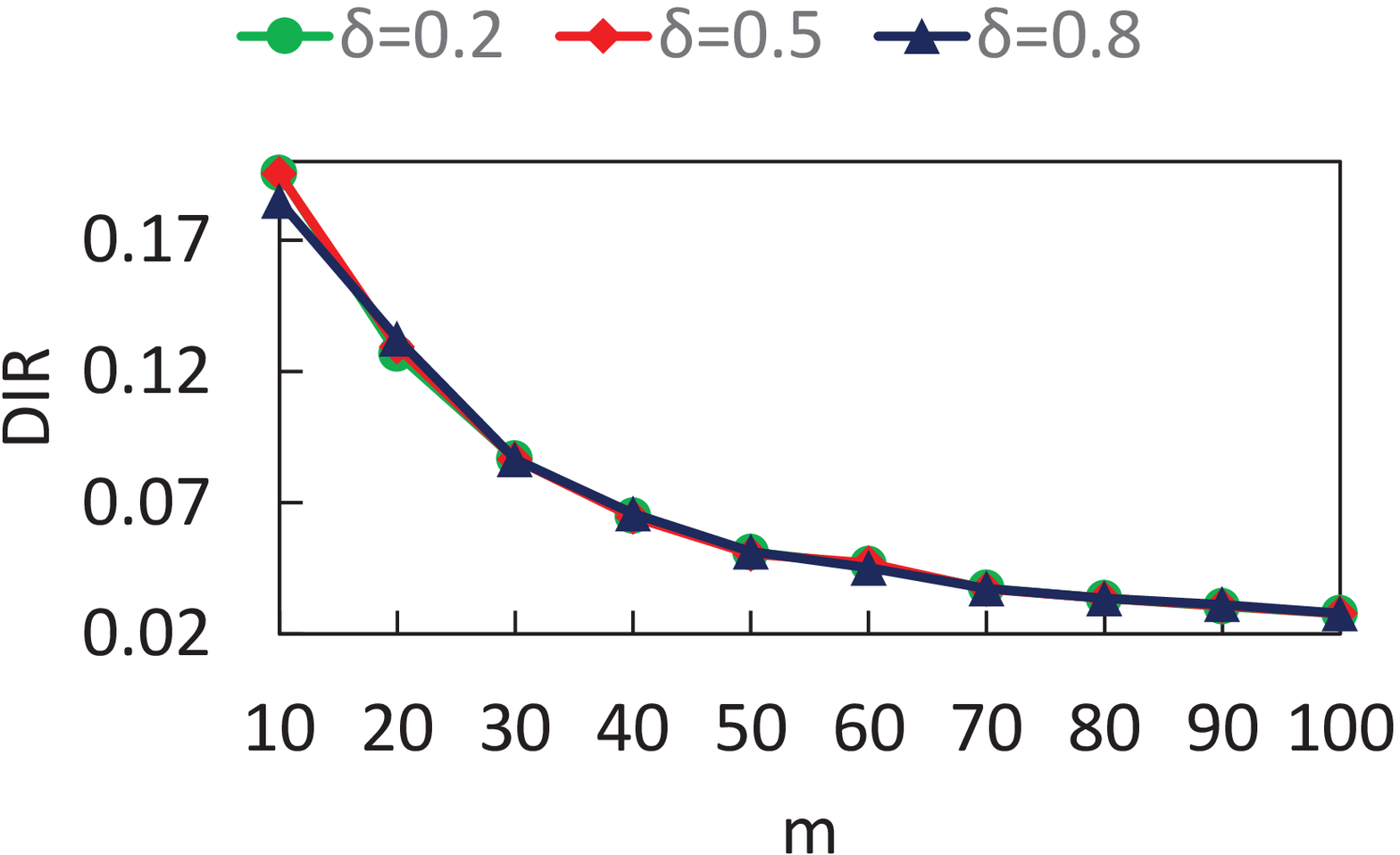}}
\quad
\subfigure[On Running Time]{
\label{double-m-rt}
\includegraphics[width=0.4\textwidth]{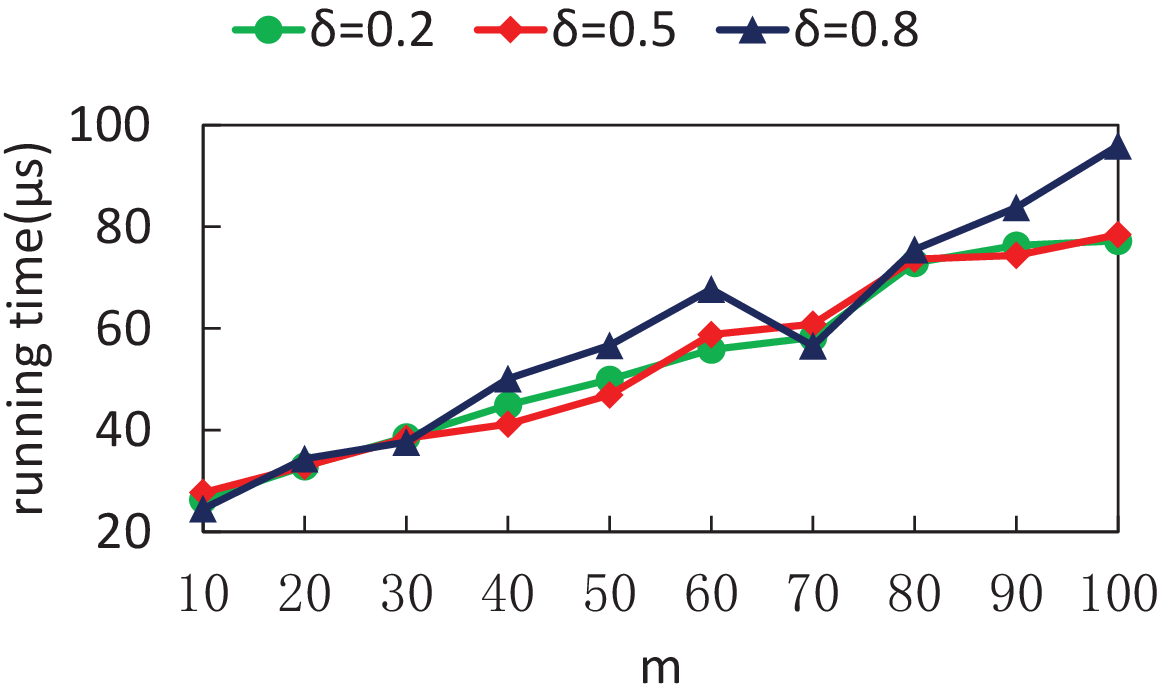}}
\caption{Impact of Memory Size $m$ (Numerical)}
\label{Figure 1}
\end{figure}	

	First, we consider the impact that $m$, i.e., the size of memory, exerts on the value of $DIR$. We present a scatter plot of DIR value vs. $m$'s value in Figure~\ref{double-m-dir}. As illustrated in the figure, it is obvious that with the increasing of $m$, DIR value tends to decline. It is all the same for various $\delta$ values. Well, the experimental results are consistent with the analysis. Because a larger $m$ means that the size of available memory gets larger, then only replacing one element into memory to increase the overall diversity, i.e., the variance of returned results stored in memory, will be harder. Therefore, the diversity increasing rate will decline significantly.
	
	Then as observed from Figure~\ref{double-m-rt}, which describes $m$'s influence on the running time of a single experiment, when $m$ becomes larger, the running time gets larger correspondingly, and it holds for various $\delta$ values. This is because with the increasing of $m$, to calculate the original diversity of the results stored in memory takes more time. Also, as for each input element $\varphi$, when considering replacing it into memory, it takes more time to compute the ``new'' diversity, i.e., variance, of the ``new'' result set. Therefore, the running time will also increase.
\subsubsection{Impact of Number of Elements To Be Scanned}
\begin{figure}
\centering
\subfigure[On DIR]{
\label{double-k-dir}
\includegraphics[width=0.4\textwidth]{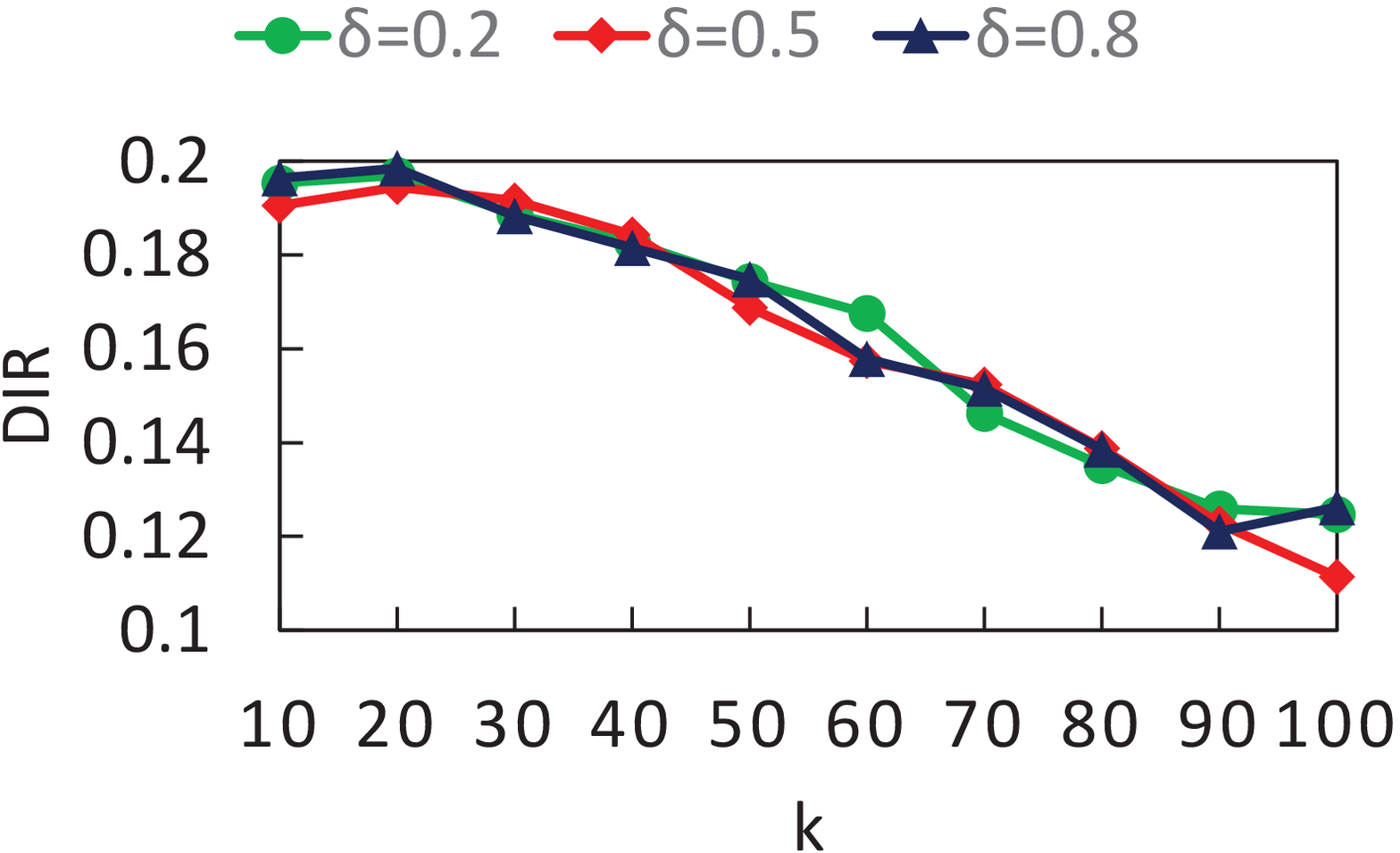}}
\quad
\subfigure[On Running Time]{
\label{double-k-rt}
\includegraphics[width=0.4\textwidth]{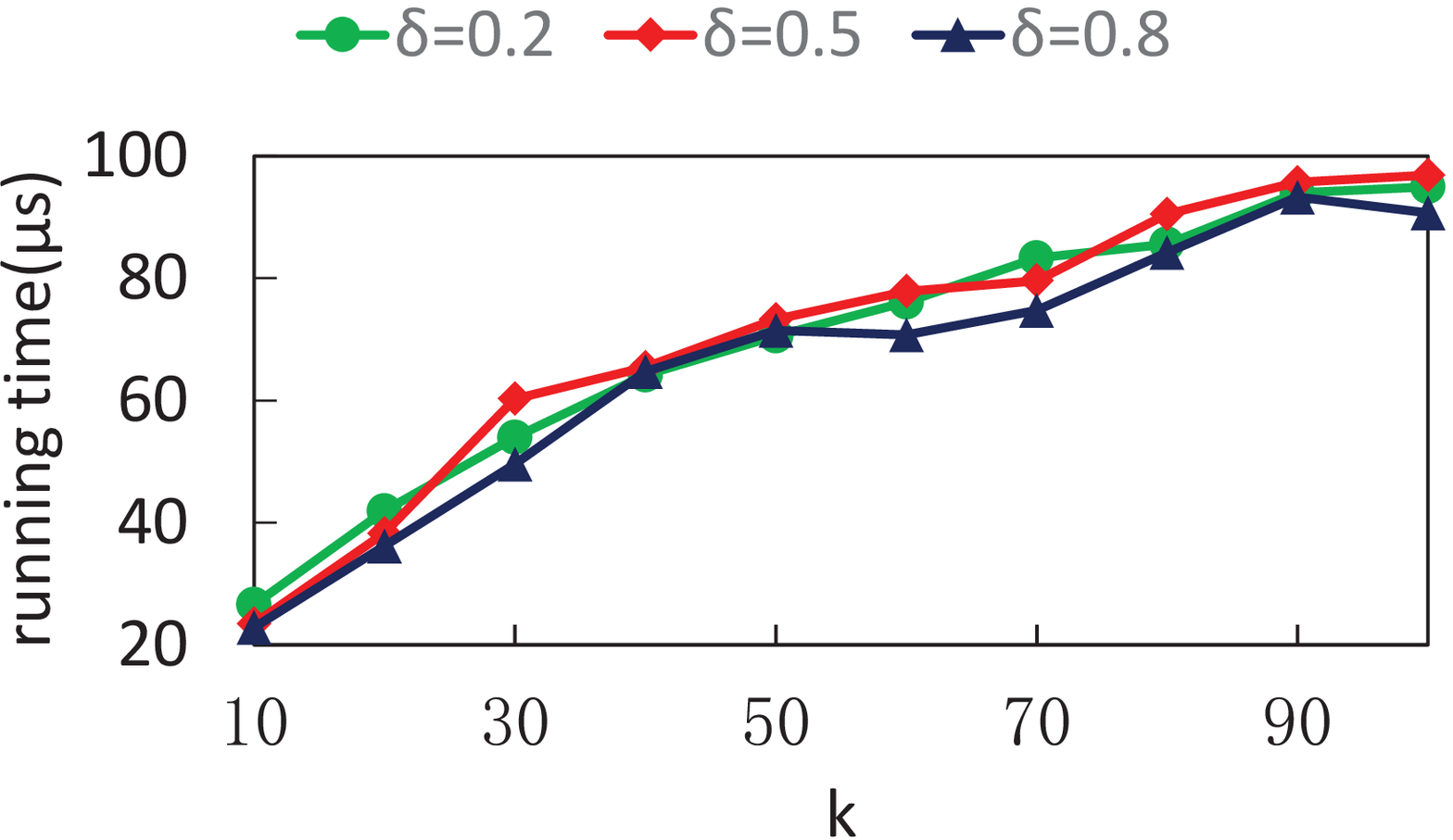}}
\caption{Impact of Scanned Elements $k$ (Numerical)}
\label{Figure 2}
\end{figure}
	In this part, the impact of the number of elements to be scanned, that is, $k$ is studied. A scatter plot of DIR value vs. $k$'s value is illustrated in Figure \ref{double-k-dir}. As observed in the figure, with $k$ increasing, the DIR value is mainly decreasing, except for a little rise in some local areas. This is true according to three different values of $\delta$, which means this experimental result holds as for various frequencies in sampling. The reason is that when $k$ gets larger, more elements are scanned to find the initial maximal PDG value. Then if the initial maximal PDG value becomes larger, the following elements will not have a larger PDG value. Thus, the last element which contributes less is selected and the DIR value gets smaller.
	
	Now we analyse the influence of $k$ on a single experiment's running time. From Figure \ref{double-k-rt}, we observe that the running time increases accordingly with the growth of $k$, regardless of $\delta$ values. This is because when $k$ becomes larger, during the procedure to find the $PDG_{\max}$ value, we have to scan more elements and thus, process more elements by computing their PDG values. Therefore, the computational overhead gets larger and then the running time of a single experiment increases.
\subsubsection{Impact of Number of Input Elements}
\begin{figure}
\centering
\subfigure[On DIR]{
\label{double-n-dir}
\includegraphics[width=0.4\textwidth]{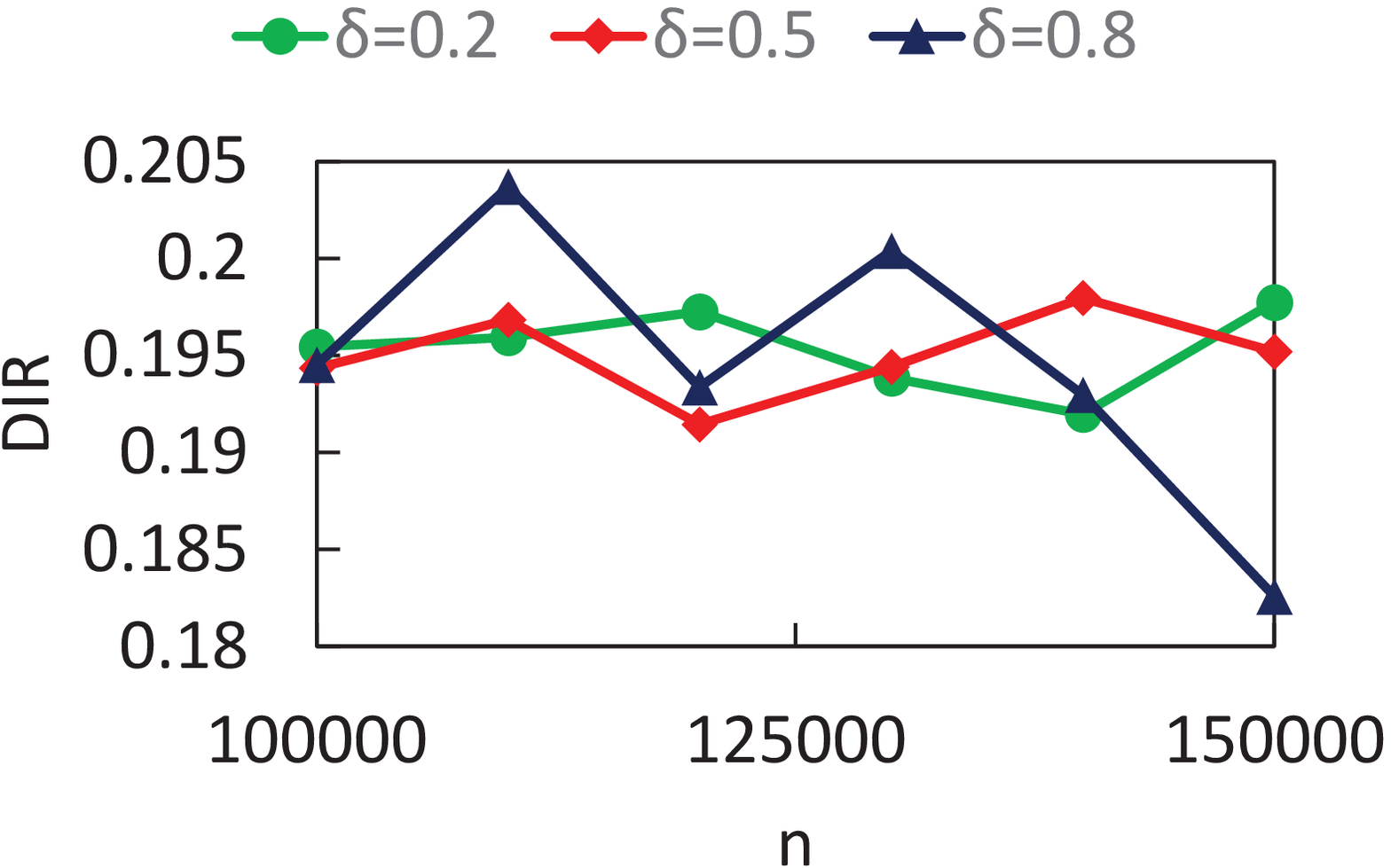}}
\quad
\subfigure[On Running Time]{
\label{double-n-rt}
\includegraphics[width=0.4\textwidth]{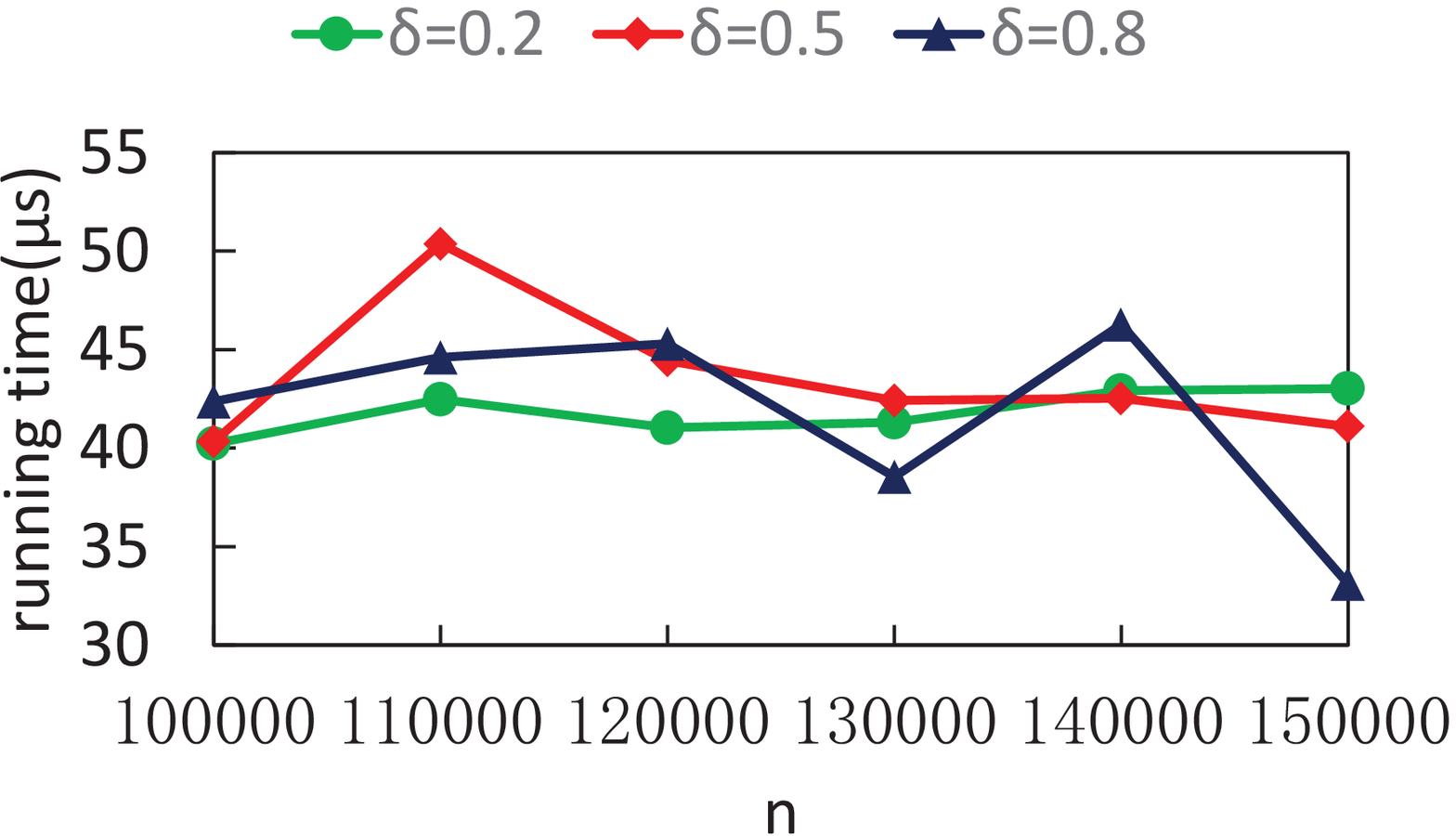}}
\caption{Impact of Input Size $n$ (Numerical)}
\label{Figure 3}
\end{figure}
	As for three $\delta$ values $0.2$, $0.5$ and $0.8$, with the input scale $n$ increasing, the corresponding DIR values are all going through several rises and falls. The plot with regard to the relationship between DIR value and $n$ is Figure~\ref{double-n-dir}. This is because when the scale of input data gets larger, the elements to be selected will be different and, meanwhile, the selected element's ability to increase the result set's variance will also be diverse.
	
	Also, the running time of a single experiment varies and fluctuates for various $\delta$ values when $n$ increases. This is illustrated in Figure~\ref{double-n-rt}. The running time is relevant to the number of elements processed in each experiment and also related to the data distribution in the selected samples. Here only the total number of input elements gets larger, but we have no idea how the data distribution is. Hence the running time may go through several rises and falls.
\subsubsection{Impact of Cardinality of Selected Samples}
\begin{figure}
\centering
\subfigure[On DIR]{
\label{double-a-dir}
\includegraphics[width=0.4\textwidth]{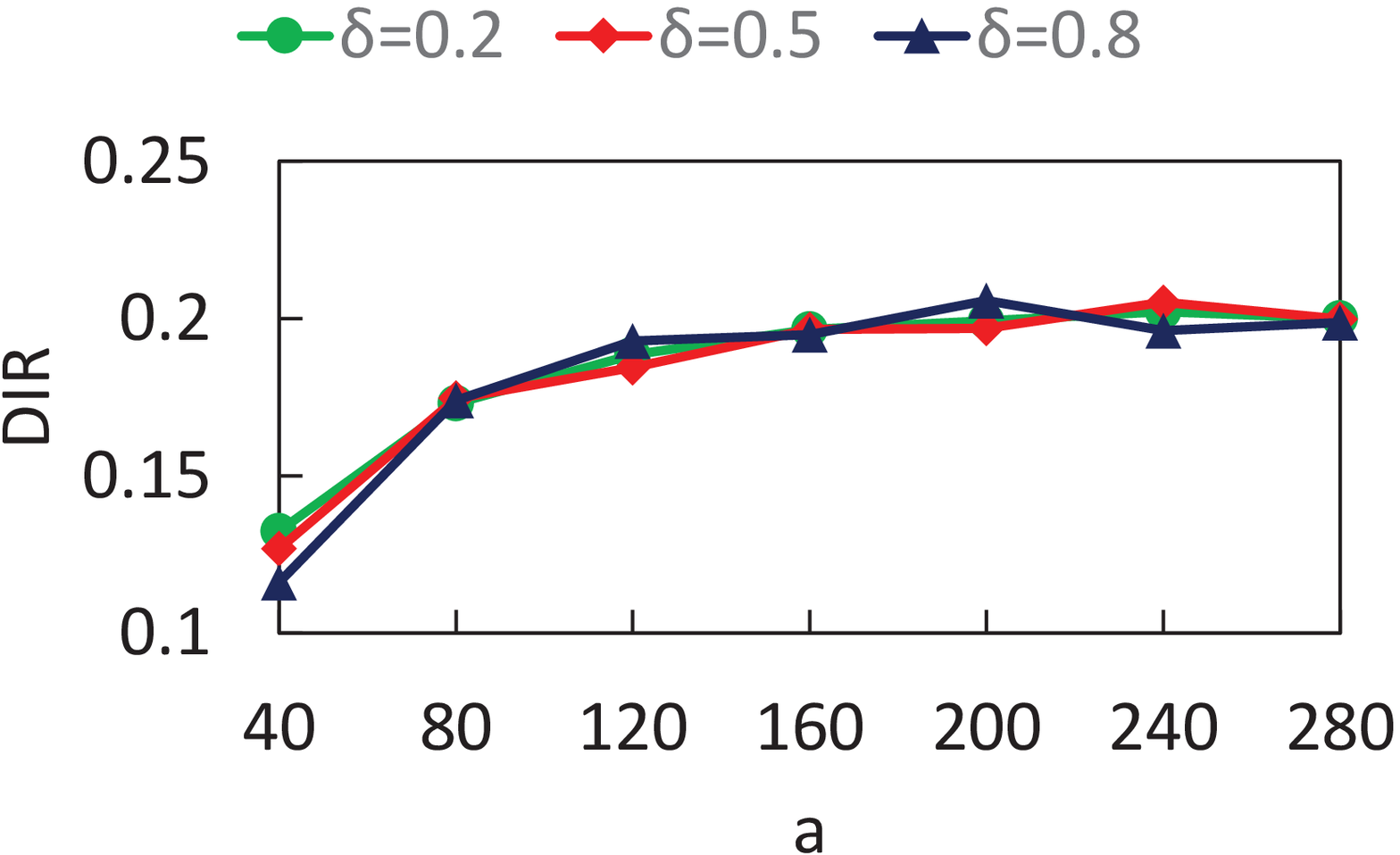}}
\quad
\subfigure[On Running Time]{
\label{double-a-rt}
\includegraphics[width=0.4\textwidth]{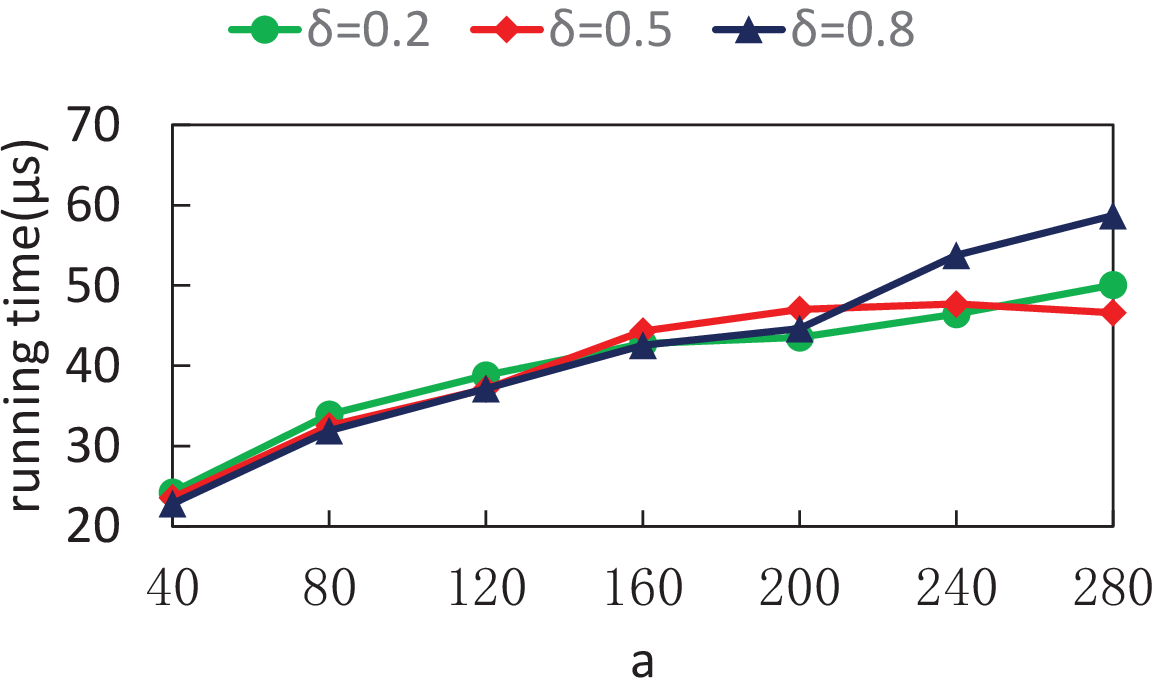}}
\caption{Impact of Sample Size $a$ (Numerical)}
\label{Figure 4}
\end{figure}
	As shown in Figure~\ref{double-a-dir}, no matter what $\delta$ is, when the sample size $a$ gets larger, the DIR becomes larger on the whole. Such experimental results are understandable in that a larger $a$ indicates a larger size of samples in one test. Then the larger the sample size is, the more possible that it is to find an element with a large enough PDG. Hence, DIR gets larger accordingly.
	
	If we consider the running time of a single experiment, we obtain the results illustrated in Figure~\ref{double-a-rt}. As for three different $\delta$ values, when the value of $a$ increases, the running time will get larger accordingly. This phenomenon is understandable in that if $a$ is larger, we have to deal with more data in each sampling experiment. When the processing time of an element does not change too much, the running time of an experiment will accordingly become longer.
\subsubsection{Impact of Number of selected Samples}
\begin{figure}
\centering
\subfigure[On DIR]{
\label{double-s-dir}
\includegraphics[width=0.4\textwidth]{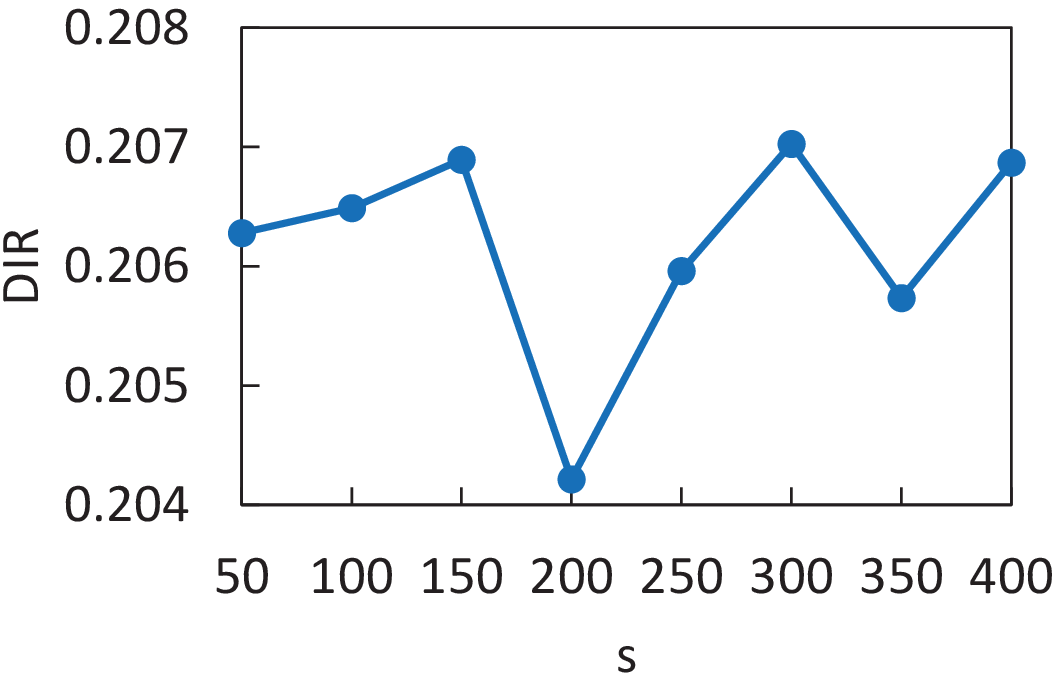}}
\quad
\subfigure[On Running Time]{
\label{double-s-rt}
\includegraphics[width=0.4\textwidth]{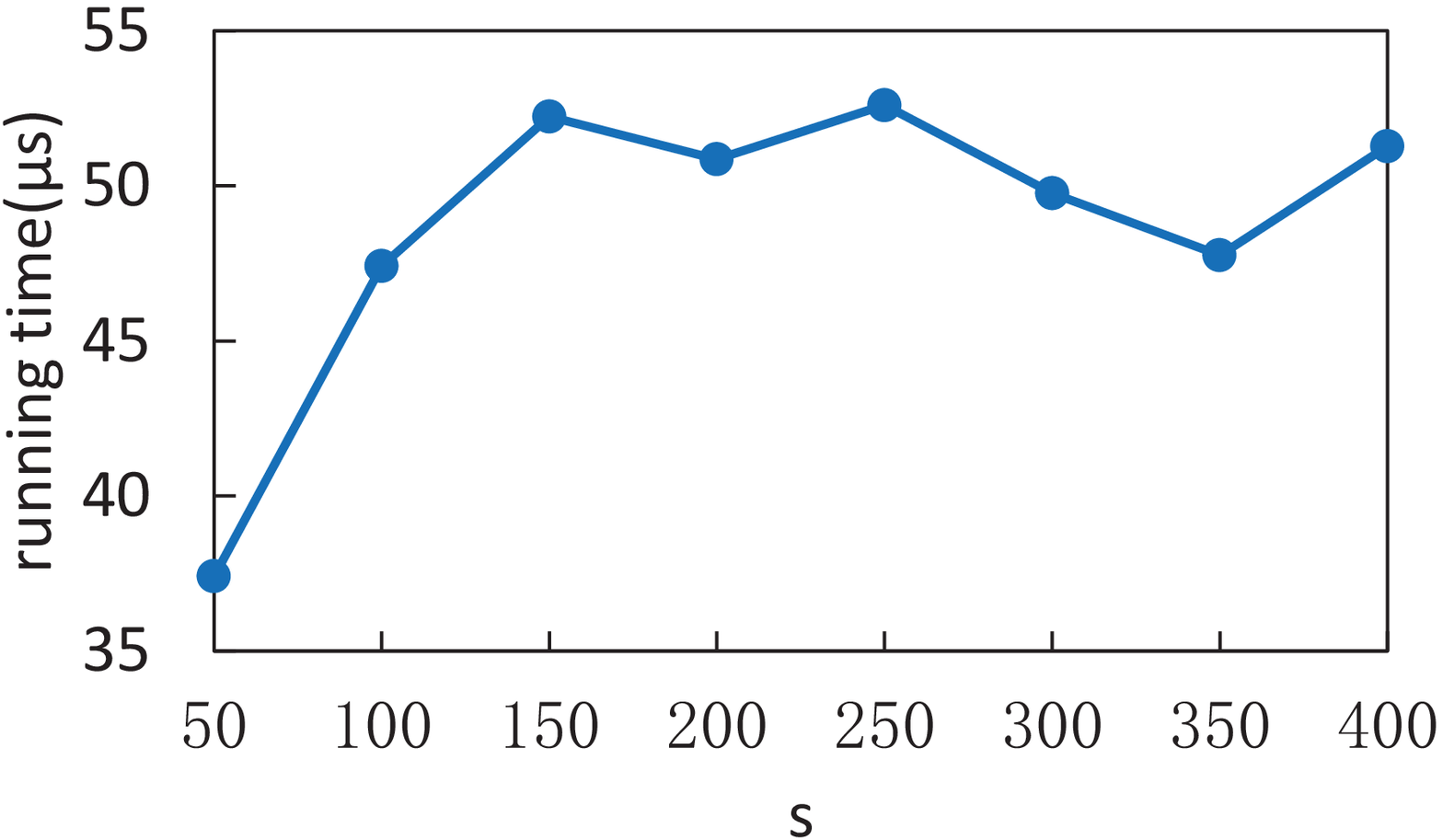}}
\caption{Impact of Sample Number $s$ (Numerical)}
\label{Figure 5}
\end{figure}	
	When we choose different numbers of samples, denoted as $s$, the DIR value varies in a fluctuating way correspondingly. The scatter plot is illustrated in Figure~\ref{double-s-dir}. This result is caused by the different values of $s$ and also, diverse selected samples.
	
	Additionally, the running time varies when $s$ changes. The scatter plot of the relationship between $n$ and the running time can be observed in Figure~\ref{double-s-rt}. The running time is also related to the data distribution of samples and which samples we choose.
\subsubsection{Discussion}
	Well, we cannot overlook the important experimental result that in Figure~\ref{Figure 1} to Figure~\ref{Figure 5}, the DIR value is all positive and in some cases, the DIR value can be quite large. This means that the result set is always more diverse after implementing the corresponding procedure of our proposed method and sometimes, the diversity is increased to a great extent. Therefore, the performance of our proposed methods to increase the diversity of results is pretty satisfactory and this proves the effectiveness of our proposed diversification method specific to numerical data.
\subsection{Experiments on String Data} \label{subsection: experiment string}
	In this section, we study how the five parameters affect the results of the proposed diversification method with regard to string data experimentally.
\subsubsection{Impact of Memory Size}
\begin{figure}
\centering
\subfigure[On DIR]{
\label{string-m-dir}
\includegraphics[width=0.4\textwidth]{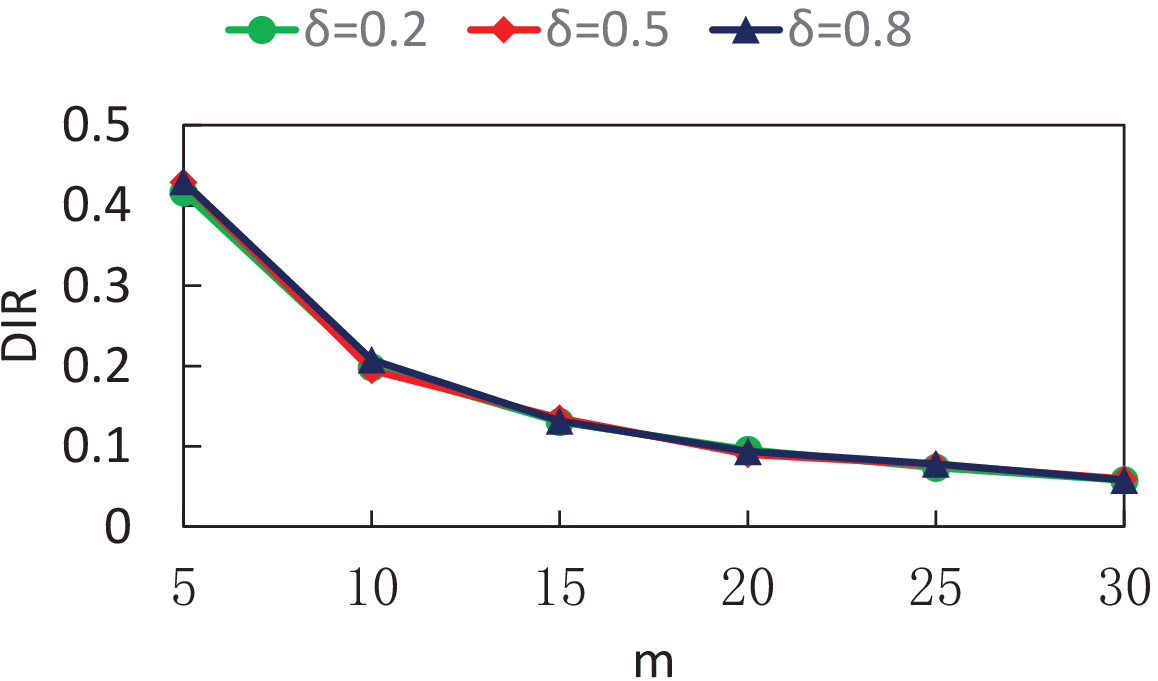}}
\quad
\subfigure[On Running Time]{
\label{string-m-rt}
\includegraphics[width=0.4\textwidth]{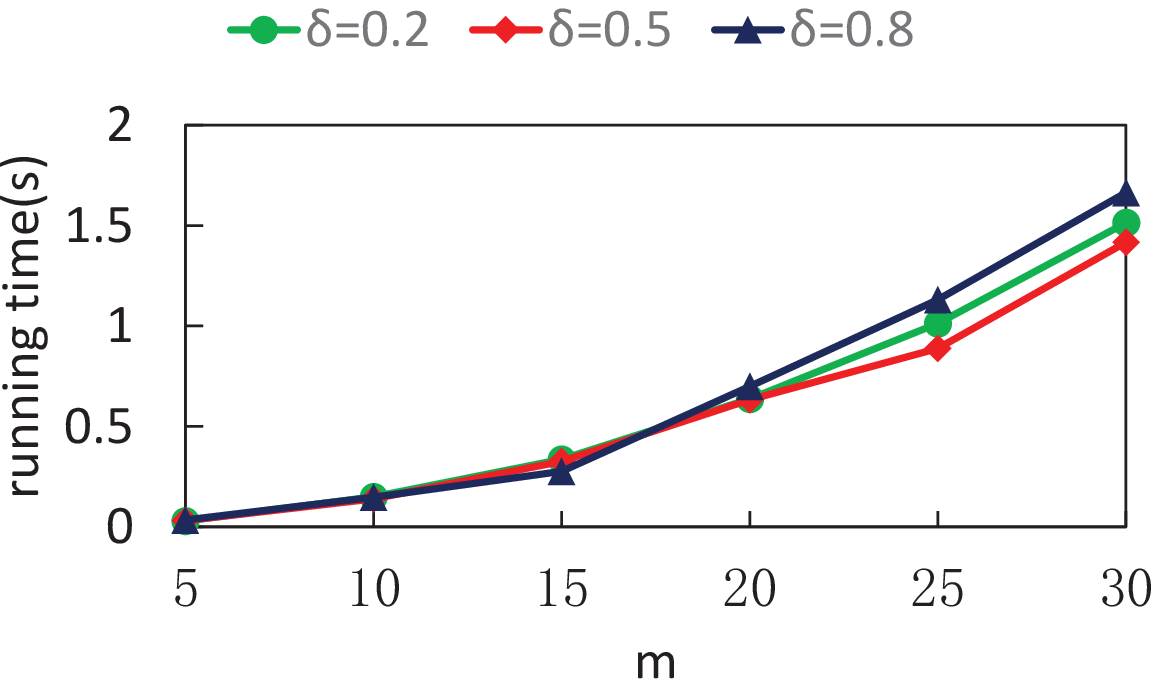}}
\caption{Impact of Memory Size $m$ (String)}
\label{Figure 6}
\end{figure}
	First, when the memory size $m$ becomes larger, the diversity increasing rate, i.e., DIR value, gets smaller and the velocity of this decline procedure gets lower. When the values of $\delta$ are various, which means that we select and use different number of samples, the results are consistent. Such results can be intuitively observed in Figure~\ref{string-m-dir}. The experimental results are in this way, since if the size of memory increases, the increase scale in diversity when a single element is replaced into memory becomes smaller. In detail, we describe the diversity of a string set as the sum of edit distances between each two elements. Then, when $m$ increases, exchanging only one element with another in memory will not contribute much to the overall diversity.
	
	Well, when $m$ gets larger, the running time of a single experiment tends to increase. This is the same no matter what the value of $\delta$ is. The relationship between the running time and the value of $m$ is intuitively illustrated in Figure~\ref{string-m-rt}. This is because when more elements are stored in memory, it takes more time to compute the original diversity, i.e., the sum of edit distances here, of the element set. Also, as for each element $\varphi$, it takes more time to compute its PDG value (demonstrated in Algorithm~\ref{Calculate-Element-PDG-2}). Therefore, the running time correspondingly increases with a larger $m$.
\subsubsection{Impact of Number of Elements to Be Scanned}
\begin{figure}
\centering
\subfigure[On DIR]{
\label{string-k-dir}
\includegraphics[width=0.4\textwidth]{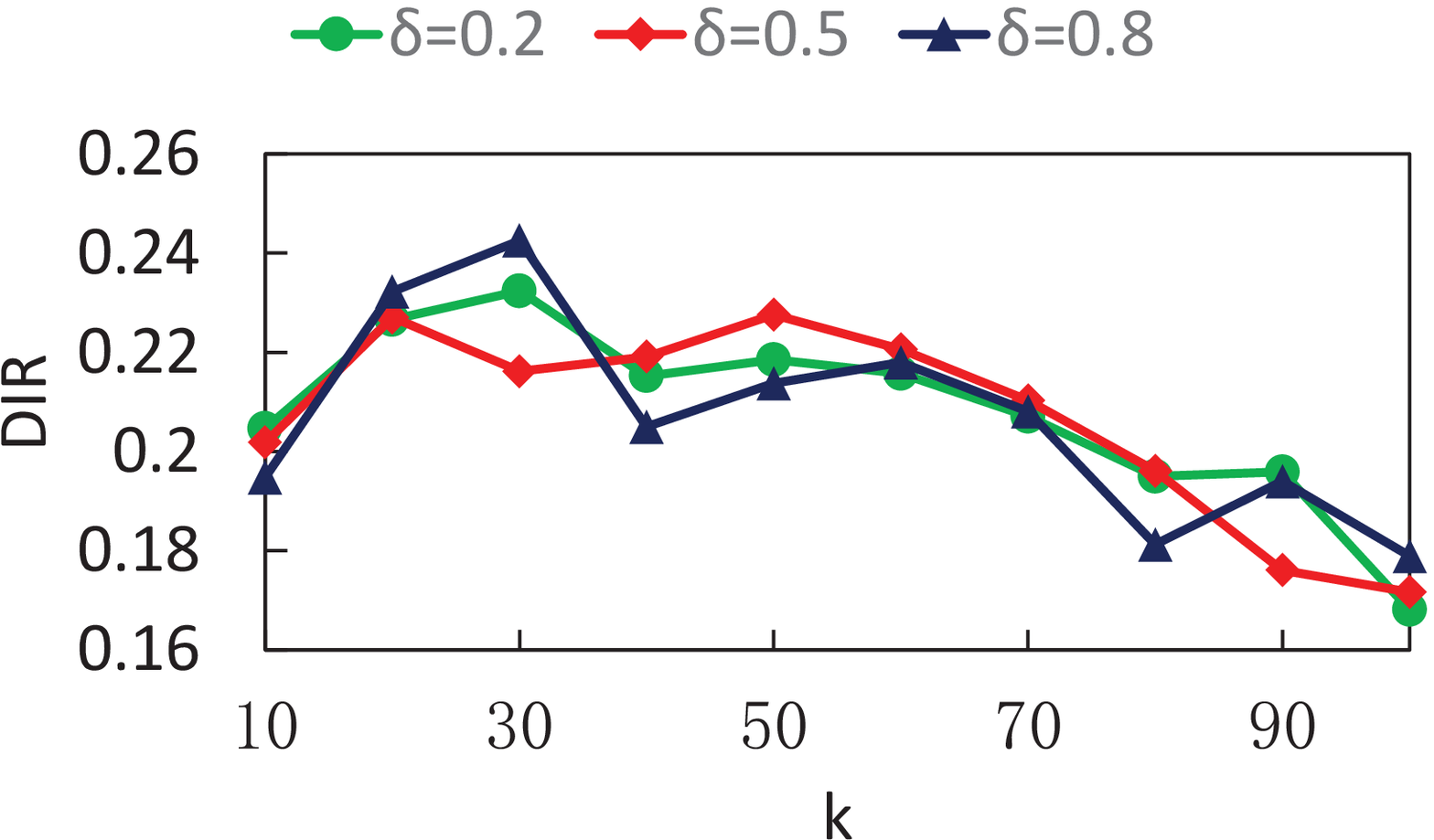}}
\quad
\subfigure[On Running Time]{
\label{string-k-rt}
\includegraphics[width=0.4\textwidth]{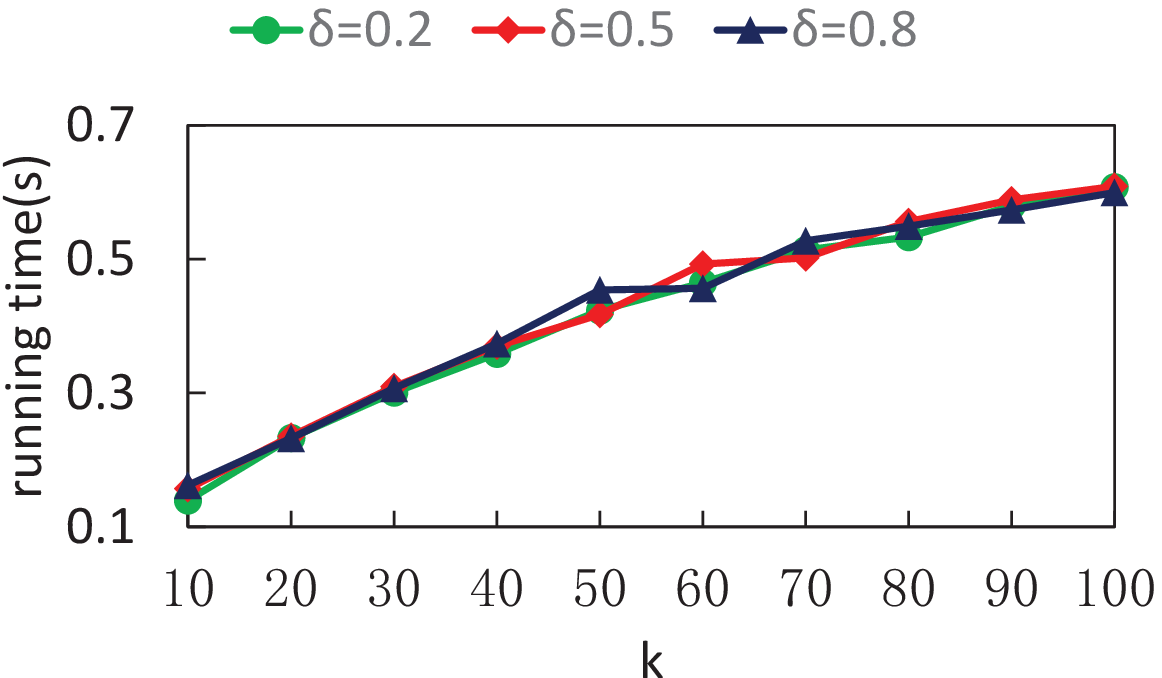}}
\caption{Impact of Scanned Elements $k$ (String)}
\label{Figure 7}
\end{figure}
	In Figure~\ref{string-k-dir}, we can easily find that when the number of elements to be scanned $k$ keeps on increasing, the DIR value first rises a little and then falls down continuously except for some fluctuations. Here whether DIR will increase or decrease is determined by the data distribution in each sample and the features of string data. If $k$ rises, and then more elements are scanned, we either can access the elements with larger PDG values in the rear part of input data or cannot find any elements with a larger PDG value than the original $PDG_{max}$. Thus, accordingly, the DIR value will either go up or fall down. Also, note that this variation tendency is the same with $\delta=0.2$, $\delta=0.5$ and $\delta=0.8$.
	
	As to the running time of a single experiment, we find that with $k$ getting larger, the running time will increase on the whole regardless of $\delta$ values, and this result can be obtained from Figure~\ref{string-k-rt}. The reason for this phenomenon is that during the procedure of our proposed diversification framework, we have to first scan the first $k$ elements to find the $PDG_{\max}$ value. When $k$ increases, we spend more time calculating each element's PDG value and pick out the maximal value among them. This is why the running time tends to become longer.
\subsubsection{Impact of Input Elements Number}
\begin{figure}
\centering
\subfigure[On DIR]{
\label{string-n-dir}
\includegraphics[width=0.4\textwidth]{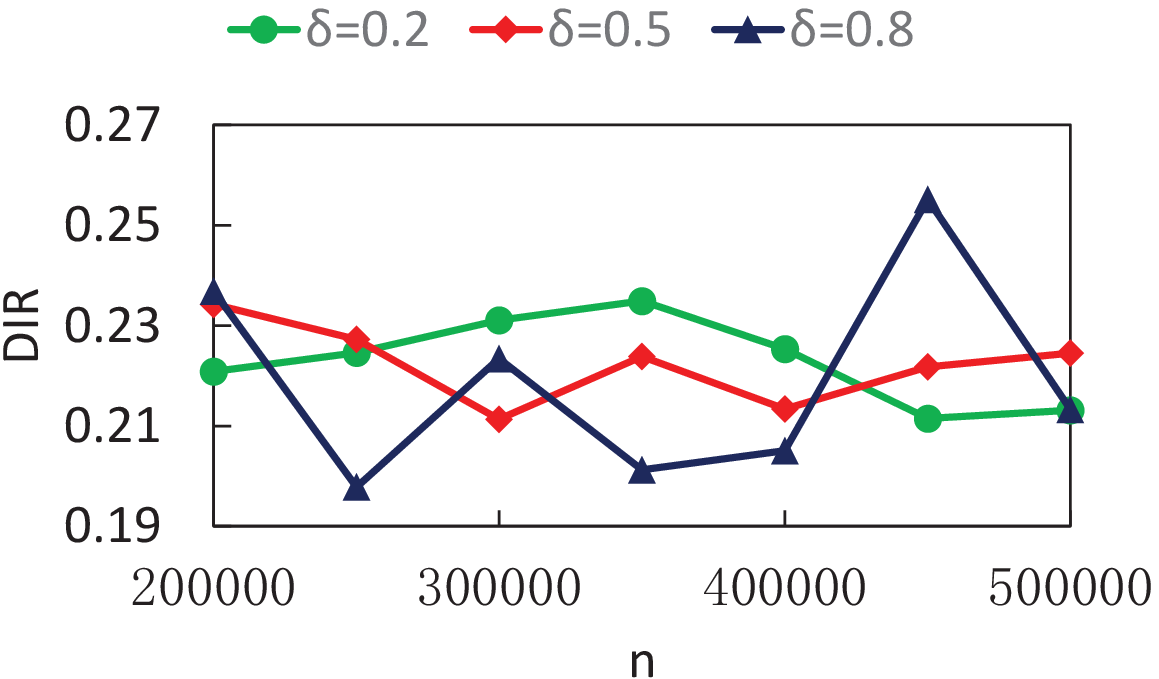}}
\quad
\subfigure[On Running Time]{
\label{string-n-rt}
\includegraphics[width=0.4\textwidth]{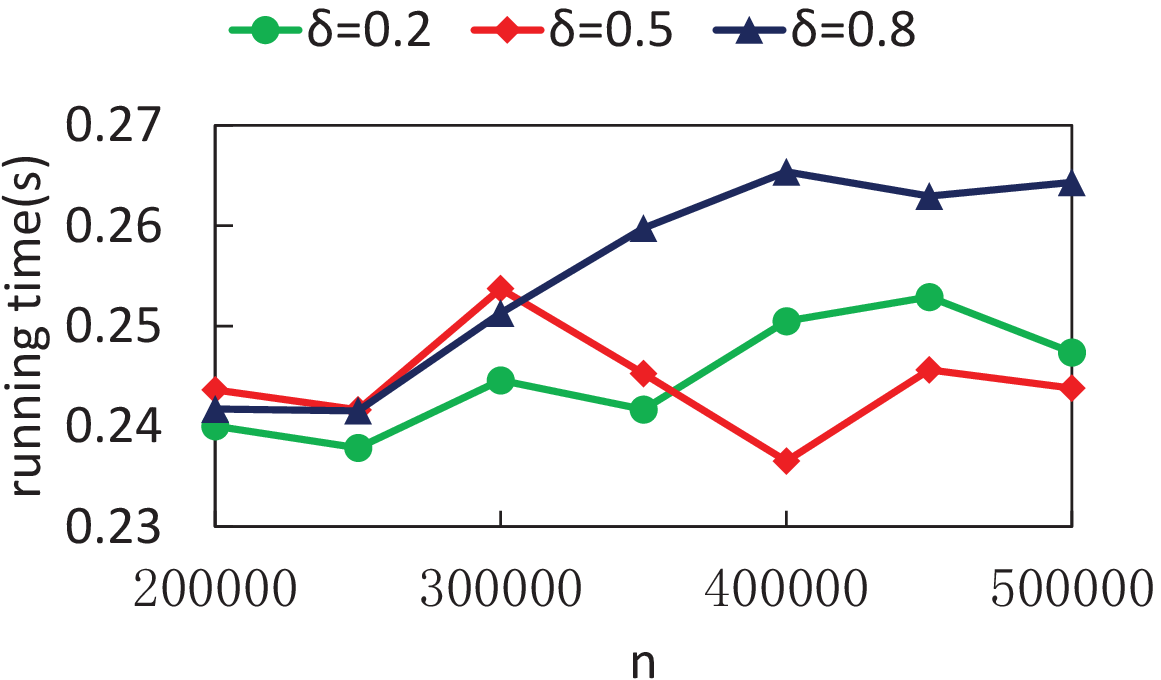}}
\caption{Impact of Input Size $n$ (String)}
\label{Figure 8}
\end{figure}
	In this section, we study the impact of input data scale $n$. As shown in Figure~\ref{string-n-dir}, similar to the experimental results specific to numerical data, here as to strings, when $n$ increases, the DIR value has several rises and falls regardless of the values of $\delta$. This is because when $n$ changes, the element that we finally choose will change accordingly. The ability of various elements to diversify results is different and hence, DIR values will fluctuate.
	
	Similarly, when the value of $n$ increases, the running time of a single experiment tend to fluctuate. The detailed change of running time is illustrated in Figure~\ref{string-n-rt}, with three different $\delta$ values. This is because the running time is determined by the exact number of elements to process and the data distribution in each selected sample. When $n$ changes, we may select various samples and thus, process different number of elements. All these factors contribute to the fluctuation of the running time.
\subsubsection{Impact of Cardinality of Selected Samples}
\begin{figure}
\centering
\subfigure[On DIR]{
\label{string-a-dir}
\includegraphics[width=0.4\textwidth]{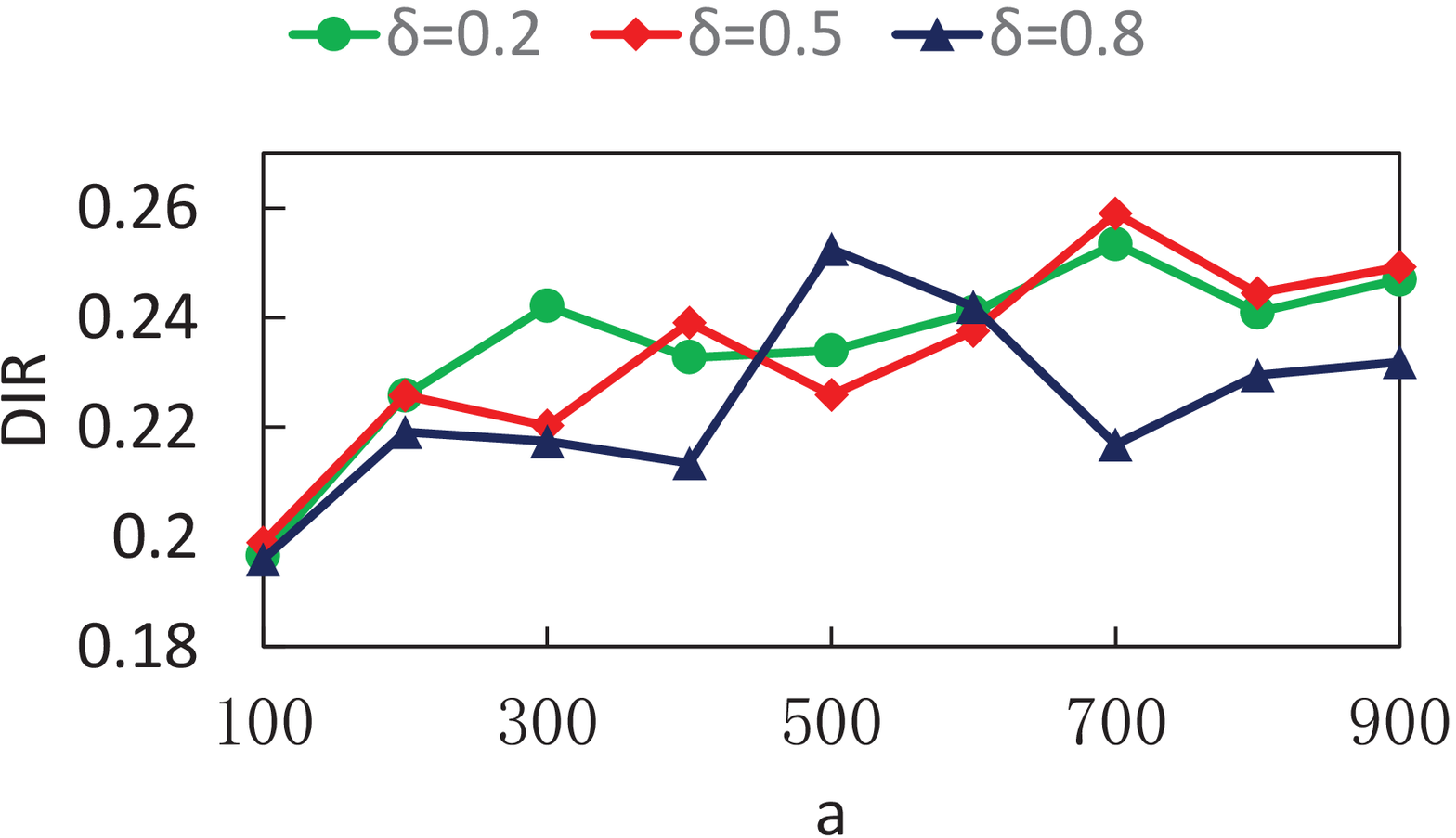}}
\quad
\subfigure[On Running Time]{
\label{string-a-rt}
\includegraphics[width=0.4\textwidth]{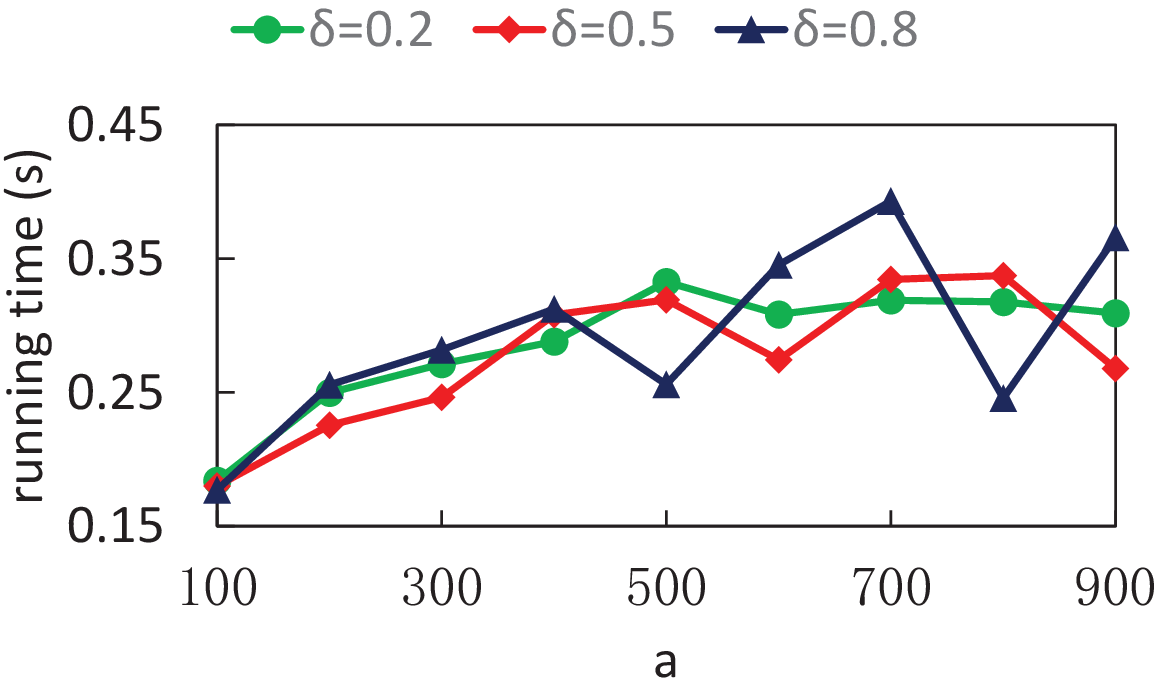}}
\caption{Impact of Sample Size $a$ (String)}
\label{Figure 9}
\end{figure}	
	Apart from all the factors discussed above, the sample size $a$ also has a significant impact on the effectiveness of our method, which is described as DIR values. Overall, the DIR value increases with the growing of $a$ except for some local fluctuations, whatever the value of $\delta$ is. This experimental result can be clearly observed in Figure~\ref{string-a-dir}. The reason is that when the sample size gets larger, the number of elements that we can deal with in a single test will become larger. Then the probability of meeting an element with a larger PDG value and promoting the growth of DIR value is larger. Therefore, when we carry out a fixed number of sampling experiments, the DIR value becomes larger.
	
	Then we pay attention to $a$'s influence on the running time of a single experiment. The change of running time with $a$ increasing can be intuitively observed in Figure~\ref{string-a-rt}. The running time mainly goes through several rises and falls, but on the whole, it tends to increase. This is because when the sample size $a$ gets larger, if the processing time of each element stays the same, an experiment's running time will get larger. However, when $a$ becomes larger, it is more likely that a following element has a PDG value larger than $PDG_{\max}$ and then terminates this exact experiment. If so, the running time will decrease. This is why the running time goes through fluctuations.
\subsubsection{Impact of Number of selected Samples}	
\begin{figure}
\centering
\subfigure[On DIR]{
\label{string-s-dir}
\includegraphics[width=0.4\textwidth]{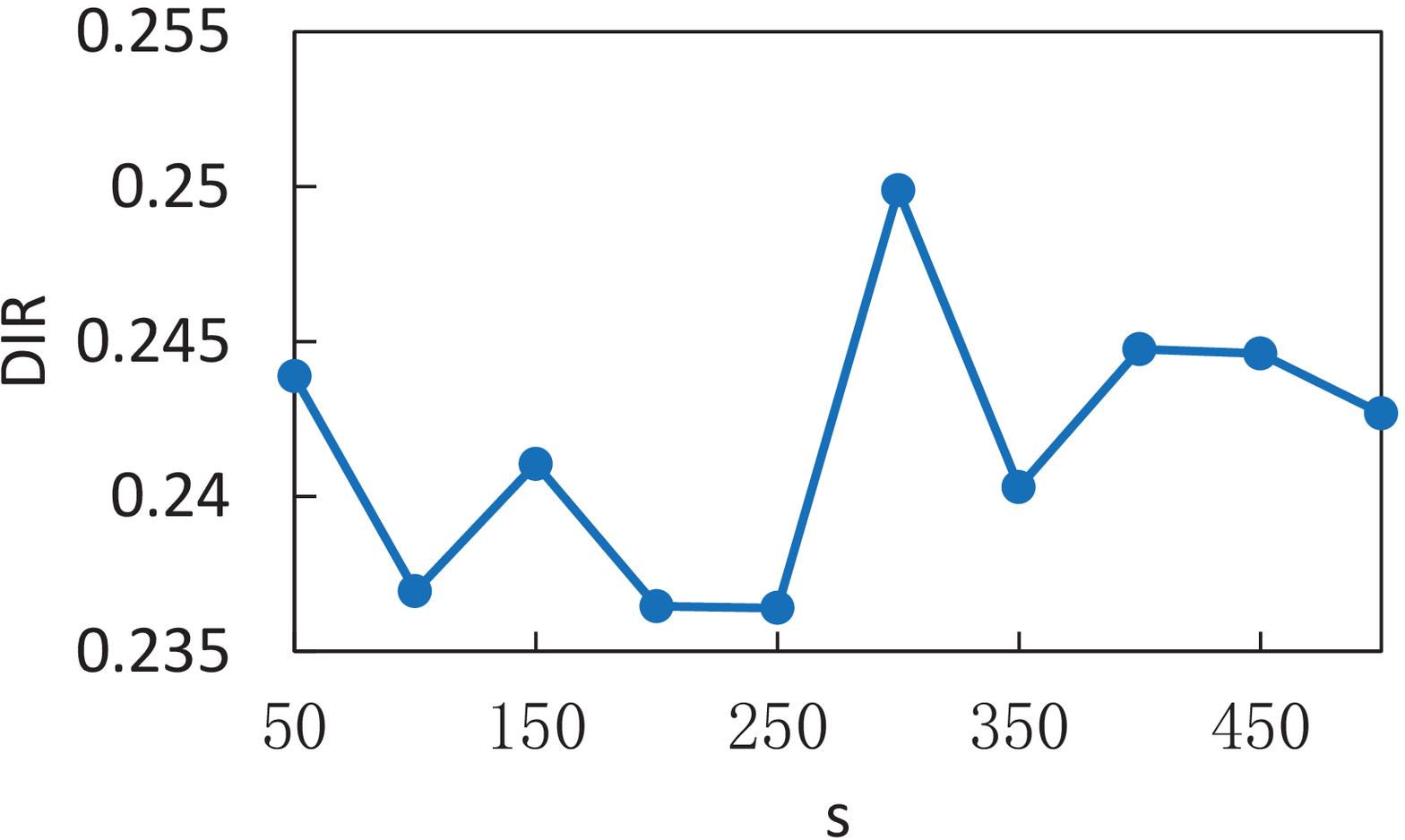}}
\quad
\subfigure[On Running Time]{
\label{string-s-rt}
\includegraphics[width=0.4\textwidth]{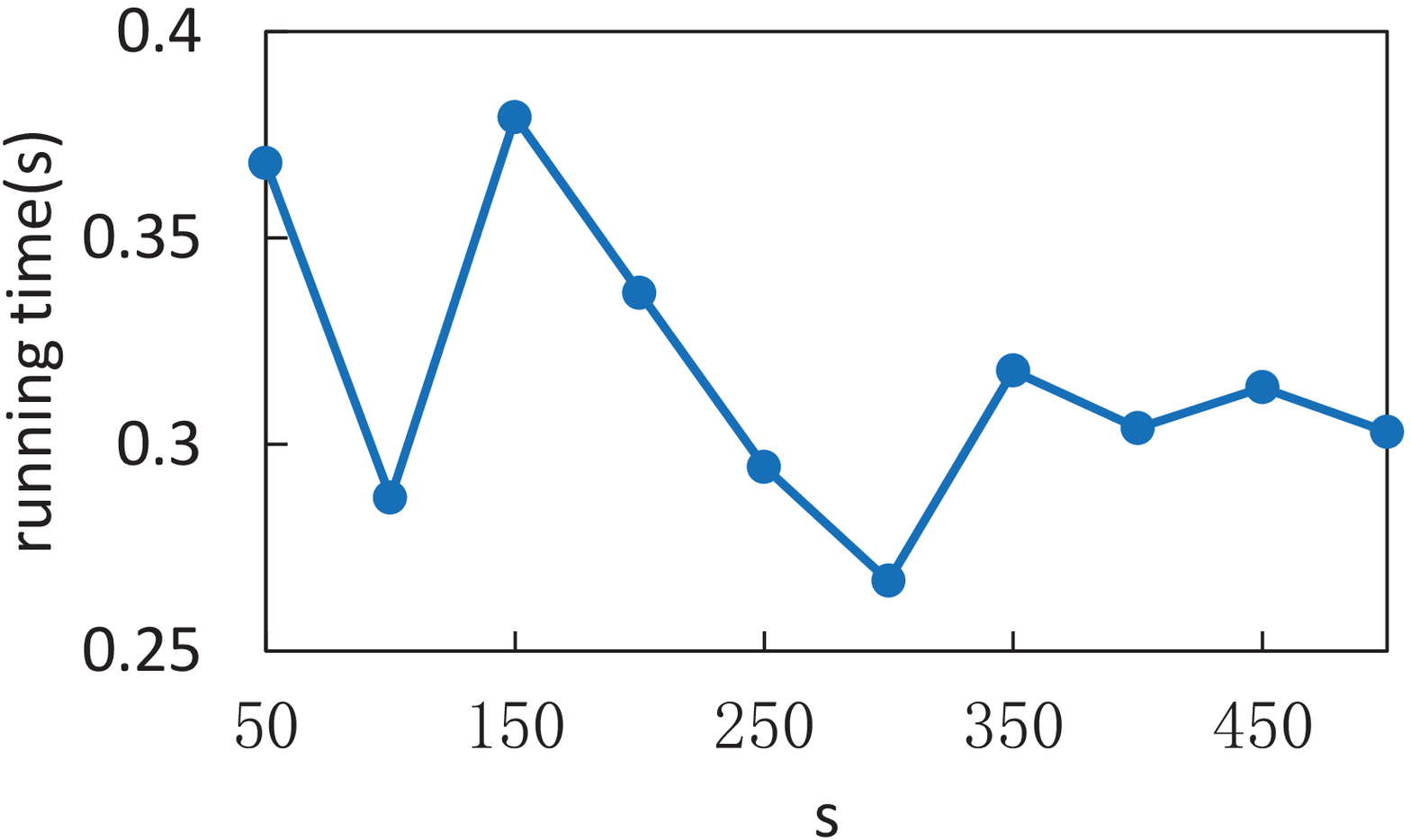}}
\caption{Impact of Sample Number $s$ (String)}
\label{Figure 10}
\end{figure}
	When we vary the number of sampling experiments executed, i.e., $s$, the DIR values will vary. This is because the number of samples and the concrete content in the samples determines which element to be selected at last and how much diverse it can increase the result set by. Figure~\ref{string-s-dir} intuitively demonstrates such experimental results.
	
	Then when $s$ tends to increase, the running time of a single experiment will fluctuate. This is also because the running time is relevant to the data distribution in the selected samples and which samples we have selected. Figure~\ref{string-s-rt} clearly illustrates the change of the running time with $s$ varying.
\subsubsection{Discussion}
	Note that in all of our sampling experiments, the diversity increasing rate is greater than zero and large enough. As the DIR value describes the degree that a set's diversity is increased, this result shows that our proposed method with regard to string data have satisfactory performance in diversifying the final result set.
	
	As thoroughly discussed before, using a large number of experiments and corresponding experimental results, we can learn that the definitions, algorithms and methods presented are meaningful and effective in the diversification procedure of both numerical data and string data. Therefore, our methods contribute much to providing users with more information and improving users' satisfaction.
\subsection{Scalability Experiments}
	In this part, we focus on the scalability of our proposed diversification framework and study its ability to deal with big data. We respectively carry out the scalability tests of both numerical data and string data, and then study and analyse them thoroughly. We also compare our approach with a MaxMin algorithm~\cite{DBLP:conf/www/GollapudiS09}, a greedy diversity query processing algorithm, whose time complexity is an approximately linear. We implement the algorithm by ourselves.

\subsubsection{Scalability Experiments of Numerical Data} \label{subsubsection: Scalability Experiments of Numerical Data}

\begin{figure}
\centering
\subfigure[On Numerical Data]{
\label{Scalability Experiments On Numerical Data}
\includegraphics[width=0.4\textwidth]{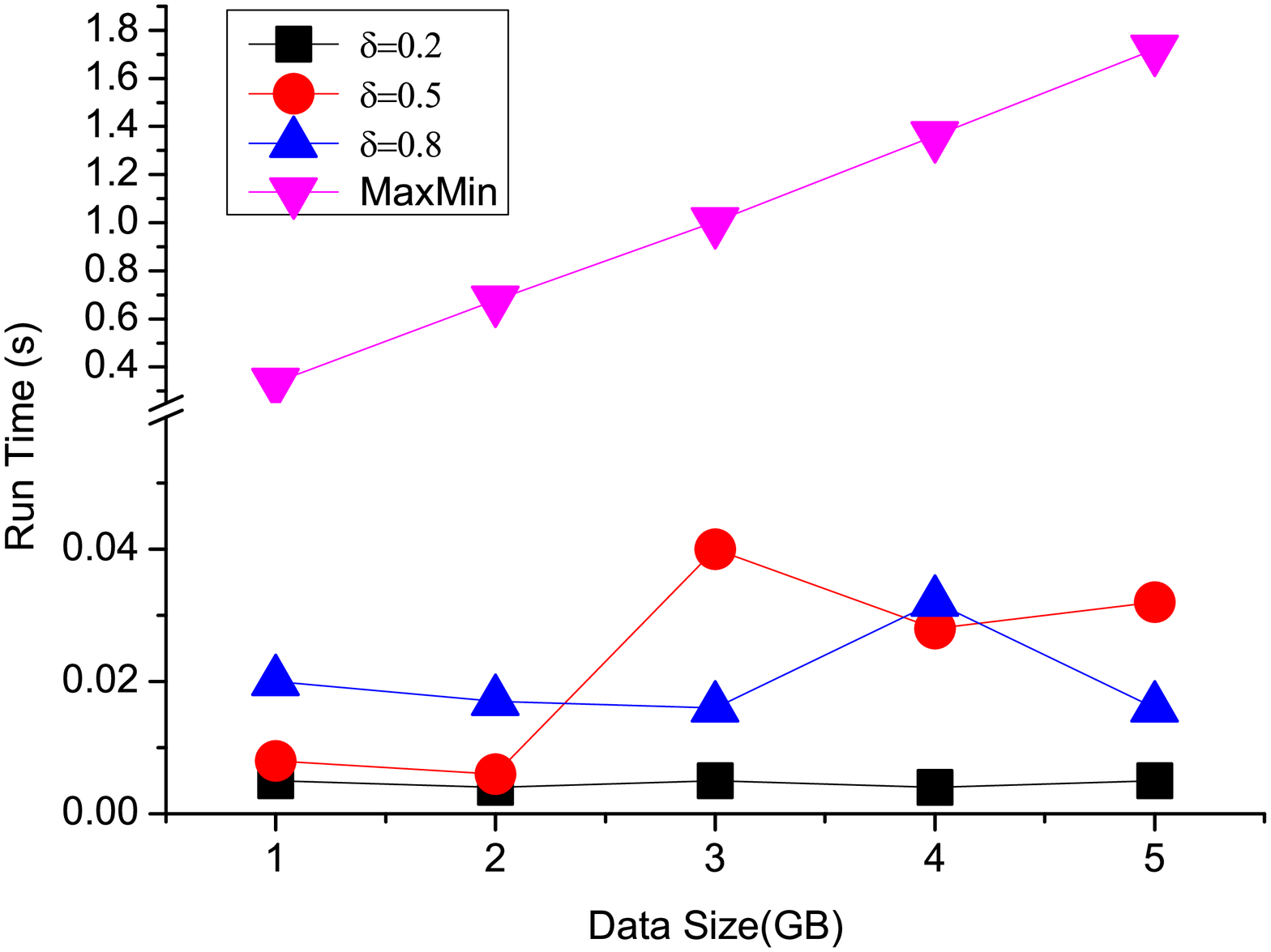}
}
\quad
\subfigure[On String Data]{
\label{Scalability Experiments On String Data}
\includegraphics[width=0.4\textwidth]{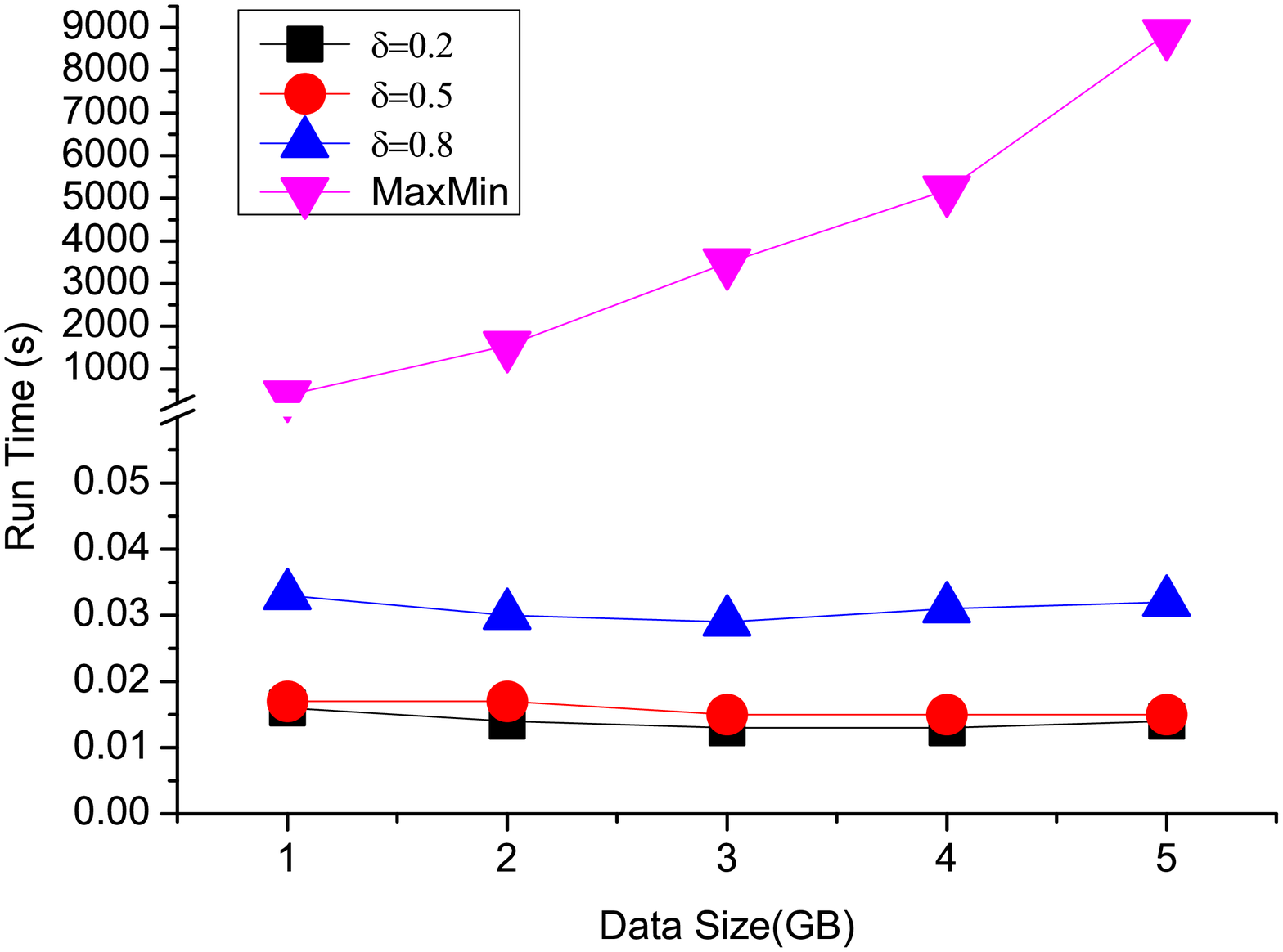}
}
\caption{Scalability Experiments}
\label{Figure 11}
\end{figure}

	When we try to carry out the scalability test of numerical data, which is in the form of double values here, we use randomly generated data which is uniformly distributed in $[-1000,1000]$. Additionally, the utilized data amount varies from 1.0 GB to 5.0 GB in order to test how our method works when dealing with big numerical data. Note that only $n$ varies here, other factors' values are set as constant values:  $m=10$, $k=20$, $a=150$. The values of $s$ and $\delta$ are set as those in Section~\ref{subsection: experiment double} and Section~\ref{subsection: experiment string}.
	
	The running time of a single experiment can well describe our proposed method's effectiveness and efficiency. The scatter plot in Figure~\ref{Scalability Experiments On Numerical Data} illustrates the relationship between the data amount and the running time. In this figure, various $\delta$ values represent different numbers of selected samples.
	
	It is intuitive from the figure that the running time is irrelevant with the data amount and the running time is less than 0.045 seconds. This means that we can always obtain the diversification results within a reasonable period of time and the performance of our proposed framework is quite satisfactory in terms of processing time when dealing with big numerical data. Therefore, our framework has a good ability to diversify big numerical data.

From comparison, our approach outperforms existing algorithm significantly. It is because we use effective sampling strategy and reduce the computation cost.
\subsubsection{Scalability Experiments of String Data}
	In this part, we analyse the scalability of our proposed framework when dealing with big data in string type. The data set is randomly generated string data which can be composed of 4 to 8 characters from $\{a,b,\cdots,z\}\cup \{A,B,\cdots,Z\}$ and follows uniform distribution. The amount of the data used ranges from 1.0 GB to 5.0 GB. In this way, we can analyse the scalability of our proposed method when processing big data in string type. Note that only $n$ changes here, so other parameters are set as constant values: $m=10$, $k=20$, and $a=200$. Also, $s$ and $\delta$ are set as the same values demonstrated in Section~\ref{subsubsection: Scalability Experiments of Numerical Data}.
	
	We also utilize the running time of a single experiment as an indicator of the performance of the proposed method and how the running times changes with the data amount is clearly shown in Figure~\ref{Scalability Experiments On String Data}.
	
	It can be seen from the figure that a single experiment's running time has nothing to do with the data amount and is no more than 0.035 seconds. This scalability test shows that our proposed method is qualified in effectiveness and efficiency when processing big data in string type. That is to say, our proposed diversification framework is capable of increasing the diversity of resulted string results when faced with big data.

From comparison, our approach is faster than MaxMin thousands of times. Our approach achieves higher performance due to our sampling strategy and avoid many costly computation of edit distance between strings.
\section{Conclusions} \label{section: conclusions}
	These years, in query processing, many diversification \linebreak[4]methods with regard to normal-scale data are proposed to meet users' information needs. However, most of the methods are not suitable in the area of big data due to big data's huge amount feature. Then how to diversify returned results when dealing with big data is well worth researching.

	In this work, we firstly propose a diversification framework which can solve the challenges that existing diversification methods are faced with when dealing with big data. This framework  processes input data online and the computational overhead as well as space overhead are low. Next, we concretely implement this framework in the area of two commonly-used data types: numerical data and string data and then design corresponding implementation algorithms. Finally, we carry out extensive experiments on real data to evaluate our proposed framework, examine the influence that $m$, $k$, $n$, $a$ and $s$ can respectively exert on the methods' performance as well as effectiveness. Additionally, scalability experiments are conducted on synthetic data to evaluate the framework's ability to process big data and demonstrate that the proposed approach outperforms the existing linear approach due to the randomized strategy.
	
	Our proposed diversification framework can only replace one element into memory to diversify the final result set. However, if we can use several caches to store elements which can be considered and replaced into memory together, we can increase the diversity of the final result set to a larger extent. This is an interesting yet challenging direction for further research.

    From the experimental results, the run time behaves linear to the memory space indicating that all elements in the memory will be scanned for each new incoming stream element. For efficiency issues, a interesting future work is to design a light weight index in the memory to avoid redundant work in subsequent PDG computations.
	
\bibliographystyle{elsarticle-num}
\bibliography{ref}

\begin{thebibliography}{10}
\expandafter\ifx\csname url\endcsname\relax
  \def\url#1{\texttt{#1}}\fi
\expandafter\ifx\csname urlprefix\endcsname\relax\def\urlprefix{URL }\fi
\expandafter\ifx\csname href\endcsname\relax
  \def\href#1#2{#2} \def\path#1{#1}\fi

\bibitem{Search_result_diversification}
M.~Drosou, E.~Pitoura, Search result diversification, SIGMOD 39~(1) (2010)
  41--47.

\bibitem{Efficient_diversity-aware_search}
A.~Angel, N.~Koudas, Efficient diversity-aware search, in: SIGMOD, ACM, 2011,
  pp. 781--792.

\bibitem{On_query_result_diversification}
M.~R. Vieira, H.~L. Razente, M.~C.~N. Barioni, M.~Hadjieleftheriou,
  D.~Srivastava, A.~Traina, V.~J. Tsotras, On query result diversification, in:
  ICDE, IEEE, 2011, pp. 1163--1174.

\bibitem{Avoiding_monotony_improving_the_diversity_of_recommendation_lists}
M.~Zhang, N.~Hurley, Avoiding monotony: improving the diversity of
  recommendation lists, in: Proceedings of the 2008 ACM conference on
  Recommender systems, ACM, 2008, pp. 123--130.

\bibitem{Diversifying_search_results}
R.~Agrawal, S.~Gollapudi, A.~Halverson, S.~Ieong, Diversifying search results,
  in: Proceedings of the Second ACM International Conference on Web Search and
  Data Mining, ACM, 2009, pp. 5--14.

\bibitem{Highlighting_Diverse_Concepts_in_Documents}
K.~Liu, E.~Terzi, T.~Grandison, Highlighting diverse concepts in documents.,
  in: SDM, SIAM, 2009, pp. 545--556.

\bibitem{Bypass_rates_reducing_query_abandonment_using_negative_inferences}
A.~Das~Sarma, S.~Gollapudi, S.~Ieong, Bypass rates: reducing query abandonment
  using negative inferences, in: Proceedings of the 14th ACM SIGKDD
  international conference on Knowledge discovery and data mining, ACM, 2008,
  pp. 177--185.

\bibitem{Novelty_and_diversity_in_information_retrieval_evaluation}
C.~L. Clarke, M.~Kolla, G.~V. Cormack, O.~Vechtomova, A.~Ashkan,
  S.~B{\"u}ttcher, I.~MacKinnon, Novelty and diversity in information retrieval
  evaluation, in: SIGIR, ACM, 2008, pp. 659--666.

\bibitem{Novelty_and_redundancy_detection_in_adaptive_filtering}
Y.~Zhang, J.~Callan, T.~Minka, Novelty and redundancy detection in adaptive
  filtering, in: SIGIR, ACM, 2002, pp. 81--88.

\bibitem{Improving_recommendation_lists_through_topic_diversification}
C.-N. Ziegler, S.~M. McNee, J.~A. Konstan, G.~Lausen, Improving recommendation
  lists through topic diversification, in: Proceedings of the 14th
  international conference on World Wide Web, ACM, 2005, pp. 22--32.

\bibitem{Improving_personalized_web_search_using_result_diversification}
F.~Radlinski, S.~Dumais, Improving personalized web search using result
  diversification, in: SIGIR, ACM, 2006, pp. 691--692.

\bibitem{Efficient_diversification_of_web_search_results}
G.~Capannini, F.~M. Nardini, R.~Perego, F.~Silvestri, Efficient diversification
  of web search results, VLDB 4~(7) (2011) 451--459.

\bibitem{It_takes_variety_to_make_a_world_diversification_in_recommender_systems}
C.~Yu, L.~Lakshmanan, S.~Amer-Yahia, It takes variety to make a world:
  diversification in recommender systems, in: Proceedings of the 12th
  international conference on extending database technology: Advances in
  database technology, ACM, 2009, pp. 368--378.

\bibitem{Efficient_computation_of_diverse_query_results}
E.~Vee, U.~Srivastava, J.~Shanmugasundaram, P.~Bhat, S.~A. Yahia, Efficient
  computation of diverse query results, in: Data Engineering, 2008. ICDE 2008.
  IEEE 24th International Conference on, IEEE, 2008, pp. 228--236.

\bibitem{A_comparison_of_p-dispersion_heuristics}
E.~Erkut, Y.~Ulk{\"u}sal, O.~Yeni{\c{c}}erio{\'g}lu, A comparison of p-dispersion
  heuristics, Location Science 4~(4).

\bibitem{Counting_distinct_elements_in_a_data_stream}
Z.~Bar-Yossef, T.~Jayram, R.~Kumar, D.~Sivakumar, L.~Trevisan, Counting
  distinct elements in a data stream, in: Randomization and Approximation
  Techniques in Computer Science, Springer, 2002, pp. 1--10.

\bibitem{Introduction_to_algorithms}
T.~H. Cormen, C.~E. Leiserson, R.~L. Rivest, C.~Stein, et~al., Introduction to
  algorithms, Vol.~2, MIT press Cambridge, 2001.

\bibitem{Probability_and_computing:_Randomized_algorithms_and_probabilistic_analysis}
M.~Mitzenmacher, E.~Upfal, Probability and computing: Randomized algorithms and
  probabilistic analysis, Cambridge University Press, 2005.

\bibitem{Efficient_Approximate_Search_on_String_Collections}
M.~Hadjieleftheriou, C.~Li,
  \href{http://dblp.uni-trier.de/db/journals/pvldb/pvldb2.html#HadjieleftheriouL09}{{Efficient
  Approximate Search on String Collections.}}, PVLDB 2~(2) (2009) 1660--1661.
\newline\urlprefix\url{http://dblp.uni-trier.de/db/journals/pvldb/pvldb2.html#HadjieleftheriouL09}

\bibitem{TPCH_double_data}
\url{http://www.tpc.org/tpch/}.

\bibitem{DBLP_string_data}
\url{http://dblp.uni-trier.de/xml/}.

\bibitem{DBLP:conf/www/GollapudiS09}
S.~Gollapudi, A.~Sharma, An axiomatic approach for result diversification, in:
  Proceedings of the 18th International Conference on World Wide Web, {WWW}
  2009, Madrid, Spain, April 20-24, 2009, 2009, pp. 381--390.

\end{thebibliography}

\end{spacing}
\end{document}